\documentclass[11pt]{article}
\usepackage[T1]{fontenc}
\usepackage{newtxtext}
\usepackage{newtxmath}   
\usepackage{setspace}
\usepackage{microtype}
\usepackage{hyphenat}

\usepackage{amsmath}

\usepackage{amsthm}

\usepackage[margin=1in]{geometry}
\usepackage{aliascnt}

\usepackage[ruled,vlined,linesnumbered]{algorithm2e}

\usepackage{array}

\usepackage{enumitem}

\usepackage{xcolor}
\usepackage{tikz}
\usetikzlibrary{calc,arrows.meta,decorations.pathreplacing,patterns}

\usepackage[hidelinks]{hyperref}
\usepackage[nameinlink,noabbrev]{cleveref}

\DeclareMathOperator{\DP}{dp}

\theoremstyle{plain}
\newtheorem{theorem}{Theorem}[section]

\newaliascnt{lemma}{theorem}
\newtheorem{lemma}[lemma]{Lemma}
\aliascntresetthe{lemma}

\theoremstyle{remark}
\newaliascnt{observation}{theorem}
\newtheorem{observation}[observation]{Observation}
\aliascntresetthe{observation}

\theoremstyle{definition}
\newaliascnt{definition}{theorem}
\newtheorem{definition}[definition]{Definition}
\aliascntresetthe{definition}
\title{An $O(n\log n)$ Algorithm for Single-Item Lot Sizing with a One-Breakpoint All-Units Discount and Non-Increasing Prices}
\author{Kleitos Papadopoulos\thanks{Email: \texttt{kleitospa@gmail.com}}}
\date{October 2025}

\setlist[itemize]{topsep=4pt, itemsep=2pt, leftmargin=1.2em}
\setlist[enumerate]{topsep=4pt, itemsep=2pt, leftmargin=1.6em}
\begin{document}
\maketitle

\begin{abstract}
This paper addresses the single-item lot sizing problem with a 1-breakpoint all-units quantity discount in a monotonic setting where the purchase prices are non-increasing over the planning horizon.
For this case, we establish several novel properties of the optimal solution and develop a hybrid dynamic programming approach that maintains a compact representation of the solution space by storing only essential information about the states and using linear equations for intermediate values. Our algorithm runs in \(O(n\log n)\) time, where \(n\) denotes the number of periods. Our result is an improvement over the previous state-of-the-art algorithm, which has an \(O(n^2)\) time complexity.
\end{abstract}

\section{Introduction}
In business, devising optimal production or ordering strategies is important for operational efficiency. Companies that aim to meet demand over a time frame while minimizing their expenses will encounter variants of the lot sizing problem. The problem requires optimizing orders, production costs and storage expenses, given the constraints of limited capacity and complex cost structures. There is extensive research in this field which signifies its importance in manufacturing efficiency.

A closely related problem is the Gas Station Problem of \cite{khuller}, which generalizes a simple variant of the lot sizing inventory problem to graphs. Some of the relevant prior work is shown in the table below.

In this paper, we tackle the single-item lot sizing problem with a 1-breakpoint, all-units quantity discount. We focus on a specific monotonic variant of the problem where unit purchase prices follow a non-increasing trend over the planning horizon and there are no setup costs. We analyze key structural properties of the optimal solution under these conditions. These properties are then used to design a hybrid dynamic programming algorithm of $O(n\log n)$ time complexity, where $n$ denotes the number of periods.

\subsection{Problem Analogy}
To provide intuition, we present our lot sizing problem using a gas station analogy and terminology where:
\begin{itemize}
  \item Time periods correspond to refueling stations
  \item Inventory corresponds to fuel in the tank
  \item Demand corresponds to the distance between stations
  \item Item cost corresponds to the fuel price
  \item \textbf{Capacity (two interpretations).}
  In the literature, capacity is commonly interpreted in one of two ways:
  \begin{itemize}
    \item \emph{Hard (tank) capacity:} the fuel level can never exceed the tank capacity, including immediately
    after refueling (post-order, pre-travel).
    \item \emph{End-of-station storage:} only the fuel level after completing the trip segment (i.e., after
    satisfying demand) is bounded; under this interpretation, the amount of fuel immediately after refueling may
    temporarily exceed the end-of-station bound.
  \end{itemize}
  In this paper we adopt the \emph{hard-capacity}. Switching to the end-of-station bound
  only changes which states are deemed feasible at each station (i.e., the truncation rule for allowable fuel
  levels) and does not affect the structure of the dynamic program or the $O(n\log n)$ time complexity of our
  algorithm.
\end{itemize}
Throughout this paper, we use both the gas station and lot sizing terminologies interchangeably as they are equivalent for our purposes.

\subsection{A brief description of the result}
A naive dynamic program (Appendix, Algorithm~\ref{alg:naive}) can generate optimal costs for every feasible inventory level at each station. These solutions viewed as \emph{solution sets} reveal a certain structure. We formalize these objects and prove structural lemmas (boundedness of a short prefix, monotone legacy-price labels, contiguity of generated points, and single-checkpoint dominance) that let us \emph{compress} whole runs of optimal states into the \emph{segments} below and update them only at a few checkpoints (termini and equal-slope boundaries). We package these checks as \emph{MV thresholds}, and use the lemmas to create our fast procedure.
These groups of states are contained in structures that we call segments where each segment has one explicit state and an equation that can be used to generate its remaining implicit states. The coefficients of these equations as well as the endpoints of the segments can change from station to station. We compute the change of the endpoints only when necessary (e.g., when an extension of a segment causes another to become suboptimal).

The segments for any set of consecutive states at any station $i$ have a specific structure which allows us to primarily use an augmented binary search tree (BST) and the standard BST operations to efficiently store and compute the optimal states for station $i+1$.

\begin{table}[!htbp]
\centering
\footnotesize
\setlength{\tabcolsep}{5pt}
\renewcommand{\arraystretch}{1.18}
\begin{tabular}{|>{\raggedright\arraybackslash}p{2.8cm}
                |>{\raggedright\arraybackslash}p{4.2cm}
                |>{\raggedright\arraybackslash}p{4.2cm}
                |>{\raggedright\arraybackslash}p{2.8cm}|}
\hline
\textbf{Reference(s)} &
\textbf{Production cost function} &
\textbf{Holding and backlogging/subcontracting cost functions} &
\textbf{Time complexity of the algorithm} \\
\hline
Chung and Lin (1988) \cite{chung1988} &
Fixed plus linear, Non-increasing setup and unit production costs, Non-decreasing capacity &
Linear holding cost &
\( O(n^2) \) \\
\hline
Federgruen and Lee (1990) \cite{federgruen1990} &
Fixed plus linear, Non-increasing setup and unit production costs, One-breakpoint all-units discount &
Linear holding cost &
\( O(n^3) \) \\
\hline
Atamtürk and Hochbaum (2001) \cite{atamturk2001} &
Non-speculative fixed plus linear, Constant capacity &
Linear holding cost, Non-speculative fixed plus linear subcontracting cost &
\( O(n^3) \) \\
\hline
Atamtürk and Hochbaum (2001) \cite{atamturk2001} &
Concave, Constant capacity &
Linear holding cost, Concave subcontracting cost &
\( O(n^5) \) \\
\hline
Mirmohammadi and Eshghi (2012) \cite{mirmohammadi2012} &
Fixed plus linear, One-breakpoint all-units discount and resale &
Linear holding cost &
\( O(n^4) \) \\
\hline
Li et al. (2012) \cite{li2012} &
Fixed plus linear, all-units discount (\( m \) stationary breakpoints) and resale &
Linear holding cost &
\( O(n^{m+3}) \) \\
\hline
Koca et al. (2014) \cite{koca2014} &
Piecewise concave with \( m \) stationary breakpoints &
Concave holding cost, Concave backlogging cost &
\( O(n^{2m+3}) \) \\
\hline
Ou (2017) \cite{ou2017} &
Piecewise linear with \( m \) stationary breakpoints &
Concave holding cost &
\( O(n^{m+2} \log n) \) \\
\hline
Malekian et al. (2021) \cite{malekian2021} &
Fixed plus linear, Non-increasing setup and unit production costs, One-breakpoint all-units discount, Constant capacity &
Linear holding cost &
\( O(n^4) \) \\
\hline
Down et al. (2021) \cite{down2021} &
Linear cost, Non-increasing unit production costs, One-breakpoint all-units discount &
Linear holding cost &
\( O(n^2) \) \\
\hline
This paper (2025) &
Linear cost, Non-increasing unit production costs, One-breakpoint all-units discount, Variable inventory capacity &
Linear holding cost &
\( O(n \log n) \) \\
\hline
\end{tabular}
\caption{Summary of relevant studies.}
\label{tab:summary}
\end{table}

\section{Organization of the Paper}
Section~3 formalizes the single-item lot-sizing model with a one-breakpoint all-units discount, non-increasing prices, and linear holding costs.
Section~4 introduces the state/segment representation and other preliminaries used throughout.
Section~5 develops the key structural observations and lemmas and introduces the MV-threshold machinery, including the holding-cost shift.
Section~6 presents the main algorithm and its analysis: an overview (6.1), the auxiliary procedures (6.2), a correctness proof (6.3), an \(O(n\log n)\) time bound (6.4), and two practical extensions—handling queries when \(B(i)>2Q\) (6.5) and decision-tracking to recover optimal plans (6.6).
Section~7 concludes. The Appendix details the augmented balanced-BST implementation (lazy tags, split/join, and related utilities) and includes the baseline dynamic programming solution for reference.

\section{Problem Definition}\label{sec:3}
The problem is about the planning of how to order inventory over a sequence of time periods (stations) \( t = 1, 2, \dots, n \). At each period:
\begin{itemize}
  \item There is a demand \( d_t \) (we also use the notation \(d(t,t{+}1)\) to describe the same thing) that must be fulfilled.
  \item You can order items (fuel) \(x_t \geq 0\) at a cost that depends on how many you order:
  \begin{itemize}
    \item If \(x_t < Q\), the unit price is \(p_{1,t}\).
    \item If \(x_t \geq Q\), the discounted unit price \(p_{2,t} \leq p_{1,t}\) applies.
  \end{itemize}
  \item Leftover items can be stored in inventory \(I_t\) with a storage capacity \(B(t)\), incurring a linear holding cost \(h_t\) per unit of inventory held at the end of period \(t\).
  \item We start with no inventory: \(I_0 = 0 \).
\end{itemize}

We also use the equivalent gas-station notation \(d(i,i{+}1)\) for the demand/distance between consecutive stations \(i\) and \(i{+}1\). One may view station \(n{+}1\) as a terminal destination with no purchasing.

The objective is to minimize the total cost of fulfilling all demands, given the pricing and inventory constraints.

More formally we can define the problem as follows:
Minimize the total procurement cost over $n$ periods, subject to meeting demands $d_t$, respecting inventory capacities $B(t)$, and using a tiered pricing function $p_t(x_t)$:
\[
\begin{aligned}
\min_{x,I}\quad & \sum_{t=1}^n \big(p_{t}(x_t)+h_t I_{t}\big),\\
\text{s.t.}\quad
& I_t = I_{t-1} + x_t - d_t,
&& t=1,\dots,n,\\
& 0\le I_t \le B(t),
&& t=1,\dots,n,\\
& I_0 = 0,\quad 0\leq x_t \le B(t) - I_{t-1}\quad B(t)\geq d_t
&& t=1,\dots,n.\\
\end{aligned}
\]
where
\[
p_{t}(x)=
\begin{cases}
p_{1,t}x, & x<Q,\\[6pt]
p_{2,t}x, & x\ge Q,
\end{cases}
\qquad t=1,2,\dots,n,
\]
and
\[
p_{1,t} \ge p_{1,t+1}
\quad \text{and} \quad
p_{2,t} \ge p_{2,t+1}.
\]

\section{Preliminary Definitions}\label{sec:4}

Central to our result is the following function:

\begin{definition}\label{def:dp}
For each time period $i$ and item (fuel) amount $f$, we define:

$\DP(i,f)$ = the minimum cost for reaching time period $i$ and having $f$ items in inventory when ready to leave time period $i$.
\end{definition}

\begin{definition}\label{def:state-tuple}
A state tuple for a specific station $i$ is a tuple containing two values, a fixed inventory level $f$ and its corresponding minimum cost $\DP(i,f)$.
\end{definition}

\begin{definition}\label{def:solution-set}
A solution set $S(i)$ for a station $i$ is a set of state tuples ordered in increasing order of their remaining fuel.
For simplicity we will refer to these state tuples in a solution set as just \emph{states} and represent them as $\DP(i,x)$ (instead of $(x,\DP(i,x))$) for a $x$ in the subsequent sections. We say that a state $\DP(i,x)$ for some $x$ is in $S(i)$ when its state tuple is contained in $S(i)$.

We distinguish between the solution sets at $i$ with states that exist before any use of fuel from $i$ which we denote as $S_b(i)$ and the $S(i)$ that contain states that are possibly using fuel from $i$ but before travelling from $i$ to $i+1$. After removing $d(i,i{+}1)$ from the remaining fuel of each state, removing any state with negative remaining fuel (that can't reach the next station) and possibly adding the holding costs, the latter is equivalent to $S_b(i{+}1)$.

A $k$-solution set for $i$ is a solution set containing an optimal state tuple (e.g., states) for each remaining fuel value between $0$ and $k$.

An approximate solution set is an optimal solution set where some states are implicit and are generated by other states (as is the case below; see \cref{def:segments,def:state-equations}). Instead of $\DP(i,f)$ being stored for all $f$ explicitly, only some entries are stored along with enough information to compute the rest. This is done in such a way that the entries that are not stored can be inferred from those that have been stored: If $\DP(i,f)$ and $\DP(i,g)$ are entries that are stored consecutively, then $\DP(i,h)$ for $f<h<g$ can be calculated from $\DP(i,f)$ using a linear equation. An approximate solution set $S(i)$ for some $i$ may be missing the first $d(i,i+1)-1$ states.

A $k$-approximate solution set is an approximate solution set that contains the states from $0$ to $k$ some of which are implicit. Note that in this case the $S(i)$ associated with a $k$-approximate solution set contains states with up to $k+d(i,i{+}1)$ remaining fuel.
\end{definition}

\begin{definition}\label{def:generate-carry}
We say a state from $S_b(i)$ \emph{generates} a state in $S(i)$ when the former is used to create the latter (e.g., using one of the equations of \cref{def:state-equations}).
A state is \emph{carried over} from $S_b(i)$ to $S(i)$ when it is the same within the two solution sets (e.g., it doesn't change from $S_b(i)$ to $S(i)$).
We also say a state $\DP(i,f)$ \emph{covers} a point $p$ (or a set of points) when it is used to generate the optimal state that has remaining fuel $p$ (e.g., the state $\DP(i,p)$).
\end{definition}

\begin{definition}[State-segments]\label{def:segments}
A state-segment (also referred to simply as a segment) in an implicit solution set is a tuple of a single state and the length of the continuous region covered by a specific state. When we say a segment covers a set of points or a region we mean its initial explicit state can be used to infer/generate the optimal states in that region.
\end{definition}

\begin{definition}\label{def:reachable}
A station $i$ (located at $d(i{-}1,i)$ from its previous station) is reachable by a segment (or a state) in $S_b(i{-}1)$ or $S(i{-}1)$ when that segment covers a region that includes a point $y$ with $y\geq d(i{-}1,i)$.
\end{definition}

\begin{definition}\label{def:segment-info}
With each state-segment for an entry $\DP(i,f)$, we maintain the following associated state information:
\begin{itemize}
  \item $p(i,f)$: the price $p_{2,i'}$ at which additional fuel could be bought at station $i'$ (where $i'$ is index of the station where the segment last received $p_2$). We also use this function for the price of the $p_2$ fuel that a state has received most recently.
  \item $r(i,f)$: the maximum amount of additional fuel that could be bought at $i'$ (limited by station capacities).
  \item Besides the above fuel which is available for each segment individually there is also the fuel that is available at each station $i$ for all segments. The \textbf{current station fuel} are the fuel types that can be used in combination with the state, namely $p_{2,i}$ and $p_{1,i}$ from the station $i$. The \textbf{optimal fuel} for a segment is the cheapest (per unit) of all the available fuels to it.
\end{itemize}
\end{definition}

\begin{definition}[State Equations]\label{def:state-equations}
The equations that generate a solution at point $f+x$ (e.g., generate $\DP(i,f+x)$) in a segment where $\DP(i,f)$ is the initial state of the segment are:
\begin{align*}
\mathrm{Eq}(x) &= \DP(i,f) + x \cdot p_{1,i}
&& \text{if } p_{1,i} \text{ is the optimal fuel}, \\
\mathrm{Eq}(x) &= \DP(i,f) + x \cdot p(i,f)
&& \text{if } p_2 \text{ fuel is optimal}, \\
\mathrm{Eq}(x) &= \DP(i,f) + r(i,f)\cdot p(i,f) + (x-r(i,f))\cdot p_{1,i}
&& \text{if } x > r(i,f).
\end{align*}
The general equation for the optimal fuel for any state $\DP(i,w)$ is denoted as $\mathrm{Eq}_{w}(x)$.
\end{definition}

We note that the price of the optimal fuel used in the equation that creates implicit states can change from station to station without the explicit state of the segment necessarily having been changed.

\begin{definition}\label{def:dominance}
We also say that:
\begin{itemize}
  \item A segment/state $S_1$ is \textbf{better than} $S_2$ at a point $p$ if the solution it gives has a lower value at that point.
  \item A segment/state $S_1$ \textbf{dominates} another segment $S_2$ when it is better for the totality of the region previously covered by $S_2$.
  \item A segment/state $S_1$ \textbf{weakly dominates} another segment $S_2$ when it is no worse on the totality of the region previously covered by $S_2$
  (i.e., it has value $\le$ at every point of that region).
  \item When domination occurs, the removal of $S_2$ and the readjustment of $S_1$ is called $S_2$'s \textbf{replacement}.
\end{itemize}
\end{definition}

\begin{definition}\label{def:terminus}
The terminus of a segment is the right endpoint of the segment when the state is first created or when it is readjusted (e.g., it replaces another state), calculated by finding the last point where the state is optimal and can use $p(i,f)$.
If a segment whose initial state is $\DP(i,f)$ covers a region larger than $r(i,f)$, then its terminus is the point $f+r(i,f)$.
\end{definition}

Note that the right endpoint of a state can change from station to station due to the availability of cheaper $p_1$ fuel, while the terminus must remain invariant unless the state is removed or its adjacent state is removed.

\begin{definition}\label{def:bulk-individual}
We say that a set of segments or states in $S_b(i)$ is \emph{bulk updated} at a station $i$ when their explicit states are given $Q$ units of $p_{2,i}$ and their implicit states are given more than $Q$ units of $p_{2,i}$ implicitly via the segment equation.

Any implicit state if made explicit before the next bulk-update will receive its $p_{2,i}$ in the sense that a separate tuple will be created that will store information about the value and the remaining fuel of that specific state. This is done only in special circumstances, such as when a segment is split and one of the new segments requires an initial state.

When we say that a state is \emph{individually updated} we mean that it is given enough $p_{1,i}$ or $p_{2,i}$ to extend to its adjacent to the right segment because it dominates it.

When we say that a set of states will be individually updated we mean that all the dominated states have been removed by individual updating.
\end{definition}

\section{Key Properties and Lemmas}\label{sec:5}

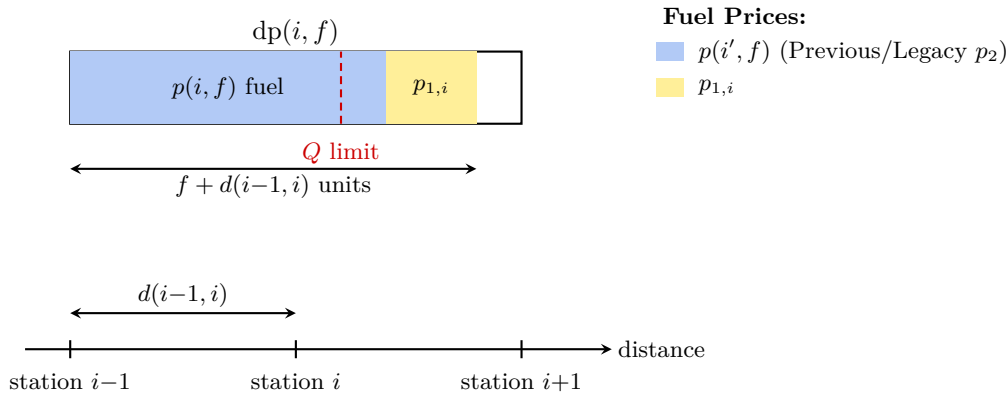
\begin{figure}[ht]
\centering
\begin{tikzpicture}[
  scale=1.2,
  font=\footnotesize,
  >=stealth,
  every node/.style={inner sep=2pt}
]
  \definecolor{bulkfuel}{RGB}{100,149,237}   
  \definecolor{regularfuel}{RGB}{255,215,0}  

  \draw[thick] (0,0) rectangle (5,0.8);

  \fill[bulkfuel!50] (0,0) rectangle (3.5,0.8);

  \fill[regularfuel!40] (3.5,0) rectangle (4.5,0.8);

  \draw[thick, red!80!black, densely dashed, line width=0.8pt]
    (3,0) -- (3,0.8);

  \node[above, font=\small\bfseries] at (2.5,0.8) {$\DP(i,f)$};

  \node at (1.75,0.4) {$p(i,f)$ fuel};
  \node at (4,0.4) {$p_{1,i-1}$};

  \node[red!80!black, below] at (3,-0.15) {$Q$ limit};

  \draw[<->, thick] (0,-0.5) -- (4.5,-0.5);
  \node[below] at (2.25,-0.5) {$f + d(i{-}1,i)$ units};

  \begin{scope}[yshift=-2.5cm]
    \draw[thick, ->] (-0.5,0) -- (6,0) node[right] {distance};

    \foreach \x/\label in {0/{$i{-}1$}, 2.5/{$i$}, 5/{$i{+}1$}} {
      \draw[thick] (\x,-0.1) -- (\x,0.1);
      \node[below] at (\x,-0.2) {station \label};
    }

    \draw[<->, thick] (0,0.4) -- (2.5,0.4);
    \node[above] at (1.25,0.4) {$d(i{-}1,i)$};
  \end{scope}

  \begin{scope}[xshift=6.5cm, yshift=0.4cm]
    \node[font=\footnotesize\bfseries, anchor=west] at (0,0.8) {Fuel Prices:};

    \fill[bulkfuel!50] (0,0.3) rectangle (0.3,0.5);
    \node[anchor=west] at (0.4,0.4) {$p(i',f)$ (Previous/Legacy $p_2$)};

    \fill[regularfuel!40] (0,-0.1) rectangle (0.3,0.1);
    \node[anchor=west] at (0.4,0) {$p_{1,i-1}$};
  \end{scope}
\end{tikzpicture}
\caption{Pictorial representation of a suboptimally created state $\DP(i,f)$ in $S_b(i)$ that was in $S(i-1)$ as $\DP(i,f+d(i-1,i))$. We note that the minimum amount of $p_2$ fuel that must be used is denoted by the $Q$-limit. Any ``excess'' additional amount may be added voluntarily but can lead to a suboptimal solution.}
\label{fig:dp-state-representation}
\end{figure}

The construction of an optimal solution set $S(i)$ can be done through the use of various algorithms.
For the sake of proving the properties and lemmas, we will assume that the construction of $S_b(i)$ for any $i$ itself can be made given $S_b(i{-}1)$ using the naive dynamic programming method presented in the appendix as Algorithm~\ref{alg:naive}.
In the following proofs the capacity $B(i)$ will be somewhat absent for the sake of simplicity but the lemmas are still correct when capacity is taken into account.

\begin{observation}\label{obs:no-refuel-if-reaches-next}
It is never necessary for any state $\DP(i,x)\in S_b(i)$ that already reaches station $i{+}1$
(i.e., $x\ge d(i,i{+}1)$) to generate a state in $S(i)$ (and hence in $S_b(i{+}1)$) by receiving
fuel from station $i$. Any state created in this way is weakly dominated by a state that postpones
the same purchase to station $i{+}1$.
\begin{proof}
Fix $\DP(i,x)\in S_b(i)$ with $x\ge d:=d(i,i{+}1)$, and suppose we buy $y>0$ units at station $i$,
thereby creating a candidate state in $S(i)$ with remaining fuel $x+y$ before traveling to $i{+}1$.
Since $x\ge d$, we can travel from $i$ to $i{+}1$ using only the original $x$ units and arrive with
$(x+y)-d=(x-d)+y$ units.

Consider instead the policy that makes \emph{no} purchase at station $i$, travels to station $i{+}1$
(using $d$ units from the original $x$), and then purchases \emph{exactly the same amount} $y$ at
station $i{+}1$. This alternative yields the same remaining fuel $(x-d)+y$ at station $i{+}1$.

The procurement cost of buying $y$ units at station $i$ is
$p_{1,i}y$ if $y<Q$ and $p_{2,i}y$ if $y\ge Q$; buying the same $y$ at station $i{+}1$ costs
$p_{1,i+1}y$ if $y<Q$ and $p_{2,i+1}y$ if $y\ge Q$. By the non-increasing price assumption,
$p_{1,i+1}\le p_{1,i}$ and $p_{2,i+1}\le p_{2,i}$, so the postponed-purchase policy is no more
expensive (and strictly cheaper if any inequality is strict). Feasibility is preserved because the
inventory level at station $i{+}1$ is identical in both policies (and holding costs, if present, only
strengthen the dominance when purchases are postponed). Hence purchasing at station $i$ from a state
that already reaches $i{+}1$ cannot be part of an optimal solution.
\end{proof}
\end{observation}
Note that this holds \emph{a fortiori} with linear holding costs. Observation~\ref{obs:no-refuel-if-reaches-next} is crucial to the analysis of cases in the next lemma since it implies that states with more remaining fuel than $d(i,i{+}1)$ are excluded from the creation of new optimal states with fuel from $i$.
\begin{lemma}\label{lem:2Q-suffices}
At any given station $i$, the optimal solution set $S_b(i)$ with the first $2Q$ states contains all the states necessary to construct the optimal $2Q$ solution set $S_b(i{+}1)$. Here we assume that $B(i)\geq 2Q+d(i,i+1)$, if it's not, then the first \( K_i := \min\{2Q,\; B(i) - d(i,i+1)\} \) states suffice.
\end{lemma}
\begin{proof}
We prove that no state $\DP(i,w)\in S_b(i)$ with $w>2Q$ is ever \emph{necessary} as a predecessor
to generate an optimal state in $S(i)$ (and hence in $S_b(i{+}1)$ after subtracting $d(i,i{+}1)$).
Equivalently, whenever an optimal state at station $i$ has fuel $f\ge 2Q$, there is an optimal
way to obtain it from some predecessor in $S_b(i)$ whose fuel is $\le 2Q$.

Suppose for a contradiction that there exists a fuel level $f=2Q+x$ and an optimal state
$\DP(i,f)\in S(i)$ such that every optimal construction of $\DP(i,f)$ uses as predecessor
a state $\DP(i,w)\in S_b(i)$ with $w>2Q$. Fix such an optimal predecessor $\DP(i,w)$.

Let $i'<i$ be the last station at which this predecessor received discounted fuel, i.e.,
its most recent $p_2$ price is $p(i,w)=p_{2,i'}$. By a standard exchange argument under the
non-increasing price assumption, we may assume (w.l.o.g.) that any non-discounted fuel units
used by the plan after station $i'$ (if any) are purchased as late as possible (so at price
$p_{1,i-1}$), and any discounted units are from the latest discounted purchase (so at price
$p_{2,i'}$); any other mixture can only be weakly improved by postponing purchases.

Let $c$ be the fuel of the state in $S_b(i')$ that is used to generate this last discounted
purchase at station $i'$. Denote by
\[
D := d(i',i)
\]
the total distance (cumulative demand) from station $i'$ to station $i$.
Since a purchase is made at $i'$, by Observation~\ref{obs:no-refuel-if-reaches-next} (applied
at station $i'$) we must have $c<D$ (otherwise the state would already reach station $i'{+}1$
and purchasing at $i'$ would be unnecessary). To arrive at station $i$ with $w$ fuel, the plan
must buy at station $i'$ an amount
\[
y \;=\; (D-c) + w
\]
so that after traveling distance $D$ we have $(c+y)-D=w$. Because this is a discounted purchase,
we have $y\ge Q$ and all $y$ units are priced at $p_{2,i'}$.

We now distinguish two cases, exactly as in the original argument.

\smallskip
\noindent\textbf{Case 1: More than $Q$ units are unavoidable.}
Assume $D\ge Q+c$, i.e., $D-c\ge Q$. Then even buying \emph{only} the minimum amount needed to
reach station $i$ from fuel $c$ requires at least $Q$ units at $i'$, so the discount at $i'$
is unavoidable. In particular, we may reduce the purchase at $i'$ by exactly $w$ units and still
remain in the discounted regime, because the reduced purchase is
\[
y-w = (D-c)\ \ge\ Q.
\]
Doing so yields a feasible plan that reaches station $i$ with \emph{zero} fuel, and its cost is
exactly $\DP(i,w)-w\,p_{2,i'}$ (same plan, minus $w$ discounted units at $i'$). Therefore,
by optimality of $\DP(i,0)$,
\begin{equation}\label{eq:dp0-bound}
\DP(i,0)\ \le\ \DP(i,w)-w\,p_{2,i'}.
\end{equation}
Now generate fuel $f=2Q+x$ at station $i$ from $\DP(i,0)$ by purchasing $f$ units at station $i$.
Since $f\ge 2Q\ge Q$, this purchase uses the discounted price $p_{2,i}$, so the resulting cost is
\[
\DP(i,0) + f\,p_{2,i}
\;\le\;
\DP(i,w)-w\,p_{2,i'} + f\,p_{2,i}
\quad\text{(by \eqref{eq:dp0-bound}).}
\]
Using $p_{2,i'}\ge p_{2,i}$ (non-increasing discounted prices), we get
\[
\DP(i,w)-w\,p_{2,i'} + f\,p_{2,i}
\;\le\;
\DP(i,w) + (f-w)\,p_{2,i}
\;\le\;
\DP(i,w) + (f-w)\,p_{j,i},
\]
where $j\in\{1,2\}$ is the tier used by the original construction from $\DP(i,w)$ to reach $f$
($j=1$ if $f-w<Q$, else $j=2$), and the last inequality holds because $p_{1,i}\ge p_{2,i}$.
Thus $\DP(i,f)$ can be obtained with no higher cost from the predecessor $\DP(i,0)$, contradicting
the assumed necessity of a predecessor with $w>2Q$.

\smallskip
\noindent\textbf{Case 2: More than $Q$ units are unnecessary.}
Assume $Q+c>D$, i.e., $D-c<Q$. Then reaching station $i$ from fuel $c$ does \emph{not} require a
$Q$-sized purchase at $i'$, so obtaining the discount at $i'$ forces us to buy \emph{extra} units.
Let
\[
s \;:=\; Q-(D-c)\ \in (0,Q]
\]
be the amount of fuel that would remain at station $i$ if we bought \emph{exactly} $Q$ units at $i'$.
Since we actually arrive with $w$ fuel, define the ``excess'' above the $Q$-threshold by
\[
h \;:=\; w-s \;>\; 0,
\qquad\text{so that}\qquad w=s+h,
\]
and note that reducing the purchase at $i'$ by $h$ units leaves a purchase of exactly $Q$ units,
which is still discounted. Hence we can reduce the purchase at $i'$ by $h$ discounted units and
obtain a feasible plan reaching station $i$ with fuel $s$, at cost $\DP(i,w)-h\,p_{2,i'}$.
Therefore, by optimality,
\begin{equation}\label{eq:dps-bound}
\DP(i,s)\ \le\ \DP(i,w)-h\,p_{2,i'}.
\end{equation}
Since $s\le Q<2Q$, the state $\DP(i,s)$ is contained in the first $2Q$ states of $S_b(i)$.

Now generate fuel $f=2Q+x$ at station $i$ from $\DP(i,s)$ by purchasing $(f-s)$ units at station $i$.
Because $f\ge 2Q$ and $s\le Q$, we have $f-s\ge Q$, so this purchase uses the discounted price $p_{2,i}$.
Thus the resulting cost is
\[
\DP(i,s) + (f-s)\,p_{2,i}
\;\le\;
\DP(i,w)-h\,p_{2,i'} + (f-s)\,p_{2,i}
\quad\text{(by \eqref{eq:dps-bound}).}
\]
Using $w=s+h$, the right-hand side becomes
\[
\DP(i,w) + (f-w)\,p_{2,i} + h\,(p_{2,i}-p_{2,i'})
\;\le\;
\DP(i,w) + (f-w)\,p_{2,i}
\;\le\;
\DP(i,w) + (f-w)\,p_{j,i},
\]
where we used $p_{2,i}\le p_{2,i'}$ and again $p_{1,i}\ge p_{2,i}$.
This contradicts the assumed necessity of a predecessor with $w>2Q$.

\smallskip
In both cases we derived a contradiction. Hence no state in $S_b(i)$ with fuel $>2Q$ is necessary
as a predecessor when forming optimal states in $S(i)$. Since $S_b(i{+}1)$ is obtained from $S(i)$
by subtracting $d(i,i{+}1)$ from the fuel coordinate, the same conclusion holds for constructing
$S_b(i{+}1)$ (in particular, for the part of $S_b(i{+}1)$ used in the next iteration).
\end{proof}

\begin{lemma}\label{lem:monotone-legacy-p2}
For any period $i$ and any two states $\DP(i,f)$ and $\DP(i,w)$ in $S_b(i)$ with $f < w$, we have $p(i,f) \geq p(i,w)$ (if the solution sets are explicit then $p(i,\cdot)$ represents the price of the $p_2$ fuel the corresponding states last received). In other words: states with less fuel have more expensive (or equal) previous $p_2$ prices.
\end{lemma}

\begin{proof}
We adopt the tie-breaking convention that, when multiple optimal states exist for the same fuel, we keep the one whose last $p_2$ purchase is as late as possible, i.e., has the smallest unit price. This preserves optimality and only affects labels $p(i,\cdot)$.

\textbf{Base case ($i=1$):} All states in $S_b(1)$ have $p(1,f) = +\infty$ (no previous $p_2$), so the property holds trivially.

\textbf{Inductive hypothesis:} Assume the property holds for $S_b(i)$: for any $f < w$, we have $p(i,f) \geq p(i,w)$.

\textbf{Inductive step:} We show the property holds for $S_b(i{+}1)$ by analyzing how states in $S(i)$ are created from $S_b(i)$.

By Lemma~\ref{lem:2Q-suffices} we may restrict attention to the first $2Q$ states of $S_b(i)$.

\smallskip
\emph{1) States carried over (no ordering at period $i$).}
Their labels $p(\cdot)$ are unchanged, so the ordering inherited from $S_b(i)$ is preserved.

\smallskip
\emph{2) States created with $p_{2,i}$ (ordering $\ge Q$ units at $i$).}
Every such state has last discounted price exactly $p_{2,i}$.
Consider, for each base state $\DP(i,u)\in S_b(i)$, the linear function
\[
\phi_u(x) \;=\; \DP(i,u) + (x-u)\,p_{2,i},\qquad x\ge u+Q,
\]
which is the cost of reaching inventory $x$ by buying at least $Q$ at period $i$ at unit price $p_{2,i}$.
All $\phi_u$ have the same slope $p_{2,i}$, so their lower envelope for large $x$ is attained by the minimum intercept
$\min_u\{\DP(i,u)-u\,p_{2,i}\}$, and thus, beyond some inventory threshold, the optimal states in $S(i)$ are exactly those created with $p_{2,i}$ (with label $p_{2,i}$). Hence the set of $p_{2,i}$-created states forms a \emph{suffix} in the inventory order of $S(i)$. Since prices are non-increasing over time, $p_{2,i}\le p(i,\cdot)$ for any carried state, so placing a constant (weakly smaller) label on a suffix preserves the nonincreasing order of labels.

\smallskip
\emph{3) States created with $p_{1,i}$ (ordering $<Q$ units at $i$).}
Fix a base state $\DP(i,g)\in S_b(i)$ and consider the affine function
\[
\psi_g(x) \;=\; \DP(i,g) + (x-g)\,p_{1,i}, \qquad g\le x<g+Q,
\]
which gives the cost of reaching $x$ using only $p_{1,i}$. All $\psi_g$ have the same slope $p_{1,i}$.
Therefore, if $\psi_g$ is optimal at some $e$ with $g\le e<g+Q$, then for any $e<e'<g+Q$ and any $x<g$,
\[
\DP(i,x)+(e'-x)p_{1,i} - \big(\DP(i,g)+(e'-g)p_{1,i}\big)
=
\Big(\DP(i,x)+(e-x)p_{1,i}\Big) - \Big(\DP(i,g)+(e-g)p_{1,i}\Big),
\]
so the sign of the difference is preserved as we move to the right within $[g,g+Q)$.
Hence the points generated from $\DP(i,g)$ using $p_{1,i}$ form a \emph{contiguous interval} to the right of $g$ (possibly empty) and they inherit the same label $p(i,g)$ as their base state. Because the base states’ labels are nonincreasing by the inductive hypothesis, concatenating these constant-label intervals preserves the nonincreasing order up to the beginning of the $p_{2,i}$ suffix.

\smallskip
Combining 1)–3), the labels in $S(i)$ are nonincreasing in inventory. Passing from $S(i)$ to $S_b(i{+}1)$ subtracts the constant demand $d(i,i{+}1)$ from inventories but does not change any labels, so the property holds for $S_b(i{+}1)$ as well.
\end{proof}

\begin{observation}\label{obs:prefix-no-skip}
First, observe that to generate an optimal state with fuel of station \( i \) for $S(i)$ we must also use an optimal state in $S_b(i)$. The latter state if its remaining fuel is equal or higher than $d(i,i{+}1)$ also generates an optimal state in $S(i)$ with no additional fuel (e.g., it exists in both $S_b(i)$ and $S(i)$). This means that no state before it in $S_b(i)$ can be used to generate a state after it in $S(i)$ (also due to the fuel of the earlier state being more or equally expensive; for this we may need to define the solution sets as always containing the leftmost possible state), this can be proven by an exchange argument.
\end{observation}

\begin{lemma}\label{lem:contiguity}
If a state $\DP(i,f) \in S_b(i)$ is used to generate optimal states in $S(i)$ using fuel of a specific price from period $i$, then these generated states form a continuous interval $[f+a, f+b]$ for some $b \geq a$ and a specific type of used fuel ($p_{1,i}$ or $ p_{2,i}$).
Furthermore, this remains true when $S_b(i)$ is stored approximately: every implicit state
$dp(i,f)$ can be viewed as a predecessor with value obtained from its segment equation, and the same contiguity conclusion follows.
\end{lemma}

\begin{proof}
Fix a period $i$. For a fuel level $u$ in $S_b(i)$, let $C(u)$ denote the optimal cost value of the corresponding state (i.e., $C(u)=\DP(i,u)$ when the state is viewed in $S_b(i)$).
When generating a state with final fuel $y$ at period $i$, a linear holding cost term $h_i\,y$ (if present) is the same for all predecessors for a fixed $y$ and therefore does not affect which predecessor is optimal; similarly, an inventory-capacity upper bound simply truncates the range of feasible $y$. Thus we omit these terms in the comparison and focus only on procurement.

We prove the claim separately for the two ways of buying at period $i$.

\smallskip
\noindent\textbf{Case 1: using $p_{1,i}$ (buying $x<Q$ units at $i$).}
From a predecessor state with fuel $u$, buying $x=y-u$ units with $0\le x < Q$ produces fuel $y$ with
\[
y\in\{u,u+1,\dots,u+Q-1\}
\qquad
(\text{or }y\in[u,u+Q)\text{ in the continuous variant}),
\]
at candidate cost
\[
C(u) + (y-u)p_{1,i}
= \big(C(u)-u\,p_{1,i}\big) + y\,p_{1,i}.
\]
Define the \emph{key}
\[
A(u) \;:=\; C(u)-u\,p_{1,i}.
\]
For a fixed target fuel $y$, the feasible predecessors $u$ for a $p_{1,i}$-purchase are exactly those with
$0\le y-u < Q$, i.e.
\[
W_y := \{u:\; y-Q+1 \le u \le y\}.
\]
Among these, the optimal $p_{1,i}$-candidate cost equals
$y\,p_{1,i} + \min_{u\in W_y} A(u)$.
To make the generating predecessor unique when there are ties, we adopt the tie-breaking rule:
\emph{among all minimizers of $A(\cdot)$ over $W_y$, choose the largest $u$}.
Let this chosen predecessor be denoted $u^*(y)$.

\smallskip
\emph{Claim: $u^*(y)$ is nondecreasing in $y$.}
Indeed, when $y$ increases by $1$, the window shifts one step to the right:
\[
W_{y+1} = \{y-Q+2,\dots,y+1\} = (W_y\setminus\{y-Q+1\})\cup\{y+1\}.
\]
If $u^*(y)\neq y-Q+1$, then $u^*(y)\in W_{y+1}$ remains feasible.
If (for contradiction) $u^*(y+1) < u^*(y)$, then $u^*(y+1)\in W_{y+1}\subseteq W_y$ and either
$A(u^*(y+1))<A(u^*(y))$ or $A(u^*(y+1))=A(u^*(y))$ with a larger index preferred at time $y$,
contradicting that $u^*(y)$ is the \emph{rightmost} minimizer of $A(\cdot)$ over $W_y$.
If $u^*(y)=y-Q+1$, then it drops out of the next window, so necessarily
$u^*(y+1)\ge y-Q+2>u^*(y)$.
Thus $u^*(y+1)\ge u^*(y)$ always.

Since $u^*(y)$ is nondecreasing in $y$, the set of target fuel values $y$ for which $u^*(y)=f$
(i.e., the optimal state at fuel $y$ is generated from $\DP(i,f)$ using $p_{1,i}$)
is a (possibly empty) set of consecutive values, hence a continuous interval.
Moreover, $f\in W_y$ is equivalent to $y\in\{f,f+1,\dots,f+Q-1\}$, so this interval has the form
$[f+a,f+b]$ with $0\le a\le b\le Q-1$ (and possibly truncated by capacity).

\smallskip
\noindent\textbf{Case 2: using $p_{2,i}$ (buying $x\ge Q$ units at $i$).}
From predecessor fuel $u$, buying $x=y-u\ge Q$ produces any target $y\ge u+Q$ at candidate cost
\[
C(u) + (y-u)p_{2,i}
= \big(C(u)-u\,p_{2,i}\big) + y\,p_{2,i}.
\]
Define
\[
B(u) \;:=\; C(u)-u\,p_{2,i}.
\]
For fixed target $y$, the feasible predecessors for a $p_{2,i}$-purchase are exactly those with $u\le y-Q$,
i.e. the prefix
\[
P_y := \{u:\; u\le y-Q\}.
\]
Thus the optimal $p_{2,i}$-candidate cost equals $y\,p_{2,i} + \min_{u\in P_y} B(u)$.
Again tie-break by choosing the largest minimizer and denote it $v^*(y)$.

\smallskip
\emph{Claim: $v^*(y)$ is nondecreasing in $y$.}
As $y$ increases, the prefix $P_y$ only expands by adding new (larger) feasible predecessors.
If $v^*(y+1)<v^*(y)$, then $v^*(y+1)\in P_y$ already, contradicting that $v^*(y)$ is the rightmost minimizer
over $P_y$. Hence $v^*(y+1)\ge v^*(y)$.

Therefore, for a fixed $f$, the set of $y$ for which the optimal state at fuel $y$ is generated
from $\DP(i,f)$ using $p_{2,i}$ is again a (possibly empty) continuous interval of fuel values.
Feasibility implies $y\ge f+Q$, so it can be written as $[f+a,f+b]$ for some $b\ge a$ (with $a\ge Q$).

This proves the first part of the lemma.

\smallskip
\noindent\textbf{Approximate solution sets (lifting to an explicit frontier).}
Assume now that $S_b(i)$ is stored as an \emph{approximate} solution set via segments.
Although only anchor states are explicitly stored, every feasible fuel level $u$ still has a well-defined
optimal value: define
\[
\widetilde C(u) := \DP(i,u),
\]
where $\widetilde C(u)$ is computed by evaluating the segment equation covering $u$
(from \cref{def:segments,def:state-equations}). Thus the approximate representation is
merely a compact encoding of the fully explicit list of values
$\{\widetilde C(u): u \text{ feasible}\}$.

Now fix a \emph{single} purchase regime at station $i$ with a fixed unit price $c$
(e.g., $c=p_{1,i}$ with the constraint $0\le y-u<Q$, or $c=p_{2,i}$ with the constraint
$y-u\ge Q$). For any predecessor fuel $u$ and target fuel $y$, the candidate cost of
reaching $y$ from $u$ under this regime is
\[
\widetilde C(u) + (y-u)c
= \big(\widetilde C(u)-u c\big) + y c.
\]
Therefore, for any fixed target $y$, choosing the best predecessor under unit price $c$
is equivalent to minimizing the \emph{key}
\[
K_c(u) := \widetilde C(u) - u c
\]
over the set of feasible predecessors for $y$.

\medskip
\emph{(i) If $c=p_{1,i}$ (buying $<Q$ units).}
The feasible predecessor set is the sliding window
\[
W_y := \{u:\; y-Q+1 \le u \le y\},
\]
exactly as in Case~1 above, except that $C(\cdot)$ is replaced by $\widetilde C(\cdot)$.
With the same tie-breaking rule (choose the \emph{largest} minimizer of $K_c$ over $W_y$),
the selected predecessor $u^*(y)$ is nondecreasing in $y$ by the same window-shift argument.
Hence, for any fixed predecessor fuel $f$, the set $\{y:\; u^*(y)=f\}$ is a (possibly empty)
interval of consecutive fuel values, i.e., $[f+a,f+b]$.

\medskip
\emph{(ii) If $c=p_{2,i}$ (buying $\ge Q$ units).}
The feasible predecessor set is the expanding prefix
\[
P_y := \{u:\; u \le y-Q\},
\]
exactly as in Case~2 above (again with $\widetilde C(\cdot)$ in place of $C(\cdot)$).
With the same tie-breaking rule, the selected predecessor $v^*(y)$ is nondecreasing in $y$,
so for any fixed predecessor fuel $f$ the set $\{y:\; v^*(y)=f\}$ is a (possibly empty)
interval.

\medskip
Thus the contiguity conclusion holds unchanged even when $S_b(i)$ is stored approximately:
implicit states behave exactly like explicit states because their values $\widetilde C(u)$ are
well-defined and the predecessor-selection problem depends only on these values and the
feasibility windows/prefixes, not on how the values are stored.\end{proof}

\textbf{Implications of Lemma~\ref{lem:contiguity}:}
This lemma implies that if a state \( \DP(i,g) \) is optimal at points \( f \) and \( p > f \) using a specific type of fuel from $i$, its segment $S_1$ must be optimal to all other points between \( f \) and \( p \). So, by this locality property, we only need to examine if the adjacent segment to $S_1$ is dominated by it before checking subsequent states in order to determine the optimal states for any $S(i)$ from $S_b(i)$. If there isn’t any state that dominates its adjacent state, that means that no state can dominate another so these states from $S_b(i)$ can be transferred to $S(i)$ unchanged.

\begin{observation}\label{obs:p1-fails-then-only-p2}
If a state in \(S(i)\) doesn’t generate any state that covers any points of the next station when \(p_{1,i}\)
(or its own $p(i,\cdot)$ if $S(i)$ is implicit) is checked for updating, then its only possibility to lead to an
optimal solution in \(S(i+1)\) is to be used with \(p_{2,i}\).
\end{observation}

\begin{lemma}[Dominance from a terminus checkpoint]\label{lem:terminus-dominance}
Let $S_1$ and $S_2$ be adjacent state-segments in $S(i)$ with explicit states
$\DP(i,f)$ and $\DP(i,g)$, where $g>f$. Let $t$ be the terminus of $S_2$.
Assume that at $t$ the state generated from $\DP(i,f)$ using fuel from period $i$
(or using its own legacy $p_2$ fuel if that is the optimal fuel) is strictly cheaper than
the state from $S_2$:
\[
C_1(t)<C_2(t).
\]
If, on the interval of $S_2$ to the \emph{left} of $t$, the \emph{unit price of the fuel used}
by $S_1$ is everywhere at least as expensive as the one used by $S_2$ (i.e., $c_1^{\leftarrow}(x)\ge c_2^{\leftarrow}(x)$
for all such $x$), then $S_1$ dominates $S_2$.

Conversely, if no other segment is optimal at the terminus $t$ of $S_2$, then $S_2$
covers at least the single point $t$.
\end{lemma}

\begin{proof}
Let $C_1(\cdot)$ and $C_2(\cdot)$ be the costs induced at period $i$ by $S_1$ and $S_2$.
Each is piecewise-affine; on any subinterval where both use fixed fuel types the slopes are
constant and equal to the \emph{unit prices of the fuel being used locally}.

\emph{Left of $t$.} Consider the maximal subinterval $I\subset[0,t]$ within $S_2$’s
left coverage on which both fuel choices remain fixed; denote their unit prices by
$c_1^{\leftarrow}$ and $c_2^{\leftarrow}$. For any $x\in I$,
\[
C_1(x)-C_2(x)
= \big(C_1(t)-C_2(t)\big) + (c_2^{\leftarrow}-c_1^{\leftarrow})(t-x)
\le C_1(t)-C_2(t) < 0,
\]
because $c_1^{\leftarrow}\ge c_2^{\leftarrow}$ by hypothesis. Hence $C_1(x)<C_2(x)$ on $I$.
If there is another kink further left, repeat the same argument on the next subinterval.
Thus $S_1$ is strictly better than $S_2$ everywhere to the left of $t$ within $S_2$’s coverage.

\emph{Right of $t$.} By definition of terminus, either (i) $S_2$ ends at $t$, or
(ii) to the right it can only use regular fuel $p_{1,i}$ (slope $p_{1,i}$). In case (i) there is nothing
to prove. In case (ii), $S_1$ can also use $p_{1,i}$ (or a cheaper legacy discounted fuel $p_2$ if available),
so its right-hand unit price $c_1^{\rightarrow}\le p_{1,i}$. For any $x\ge t$ in $S_2$’s right
extension,
\[
C_1(x)-C_2(x)=\big(C_1(t)-C_2(t)\big) + (c_1^{\rightarrow}-p_{1,i})(x-t)\le C_1(t)-C_2(t) < 0,
\]
so $C_1(x)<C_2(x)$ there as well.

Combining both sides, $S_1$ is strictly cheaper than $S_2$ on all points covered by $S_2$,
so $S_1$ dominates $S_2$.

\emph{Conversely.} If no other segment is optimal at $t$, then $S_2$ itself attains the
minimum at $t$, hence it covers at least that point; by the segment’s piecewise-affine
construction it covers a nonempty interval containing $t$.
\end{proof}

\begin{observation}\label{obs:partial-overlap}
A segment $S_1$ can cover a region that includes points from its next segment without dominating it.
In that case it cannot cover the second segment’s terminus point.
Our algorithm does not always resolve partial coverings via updates since it is computationally costly.
Since the region in question is always bounded by the two segments, we can calculate any point in that region
by generating the value at that point using both segments and taking the minimum.
\end{observation}

\paragraph{Implication of \cref{lem:terminus-dominance}.}
It suffices to check whether a state is better at the terminus of another state to know if the second segment is dominated.

\begin{lemma}[Same-slope dominance]\label{lem:same-slope-dominance}
If a state $\DP(i,f)$ from segment $S_1$ (extended to point $w\ge f$) generates a cheaper state than
$\DP(i,w)$ of segment $S_2$ at point $w$, and the optimal fuel of both segments has the same price per unit
throughout the entire interval of $S_2$, then $S_1$ dominates $S_2$.
\end{lemma}

\begin{proof}
Let the common unit price be $c$. For any subsequent point $w+p$ (with $p\ge 0$), the values produced by the
two segment equations are
\[
\mathrm{Eq}_{f}(w) + c\cdot p
\quad\text{and}\quad
\mathrm{Eq}_{w}(w) + c\cdot p.
\]
Since $\mathrm{Eq}_{f}(w) < \mathrm{Eq}_{w}(w)$ by assumption, the inequality persists for all $p\ge 0$,
so $S_1$ is cheaper than $S_2$ throughout $S_2$’s interval.
\end{proof}

\paragraph{Implications of \cref{lem:same-slope-dominance}.}
This applies when two adjacent state segments have fuel of the same price available to them and then a new cheaper fuel becomes available to both of them so that the first dominates the second.
The other case that this can be applied is when the first segment has more expensive fuel than the second and then a new cheaper fuel becomes available to both of them so that the first dominates the second.
This isn’t the case when the fuel of the second state is less expensive than the optimal fuel of the first state. In such a case, which happens when one of the two states has received cheaper $p_2$ fuel in a bulk-update, the terminus must be used.

\paragraph{On Bulk-updates.}
As it will be shown in the next sections all the segments in $S_b(i)$ whose current right endpoint lies below $d(i,i+1)$ are bulk-updated. These segments receive initially $Q$ units of $p_{2,i}$ and their internal $p_2$ fuel becomes $p_{2,i}$.
Some of these segments may overlap with other segments in the range $[Q,2Q]$. In order to determine which segments are optimal and remove the rest efficiently, we will use a procedure that uses the following lemma.

\begin{lemma}\label{lem:linemerge-band}
When creating $S(i)$ from $S_b(i)$, let $L_1$ be the set of state segments whose coverage
intersects the band $[Q,2Q]$ and that were created at earlier stations (so within $[Q,2Q]$
their unit price is either legacy $p(i,\cdot)$ or $p_{1,i}$), and let $L_2$ be the
set of newly bulk-updated segments at station $i$ (so within $[Q,2Q]$ they use $p_{2,i}$).
Then exactly one of the following holds:
\begin{enumerate}[label=\textbf{(\arabic*)},leftmargin=*]
\item \textbf{Complete takeover.} If the lower envelope of $L_2$ is no higher than the
lower envelope of $L_1$ at some point $x$ in the region covered by $L_1\cap[Q,2Q]$, then
for all $x'\ge x$ in that region the lower envelope of $L_2$ is no higher than that of $L_1$.
Consequently, all segments of $L_1$ to the right of $x$ are dominated by $L_2$(e.g $L_2$ provides no worse canditate segments in that region).

\item \textbf{No takeover inside (tail extension).} If the lower envelope of $L_2$ is strictly
above that of $L_1$ at every point of $L_1\cap[Q,2Q]$, then all segments of $L_2$ are
no-better throughout $L_1$’s coverage in $[Q,2Q]$. However, the rightmost $L_2$ segment may
become optimal immediately to the \emph{right} of $L_1$’s band—e.g., on $(2Q,\,2Q+d(i,i{+}1)]$.
\end{enumerate}
\end{lemma}

\begin{proof}
By the non-increasing price assumption, within $[Q,2Q]$ every segment in $L_2$ uses unit
price $p_{2,i}$, while every segment in $L_1$ uses a unit price $\ge p_{2,i}$ (either a
legacy $p(i,\cdot)$ or $p_{1,i}$ due to the monotonic nature of the costs defined by the problem). Thus, on $[Q,2Q]$, \emph{all} lines from $L_2$ have slope
$p_{2,i}$, and \emph{every} line contributing to the lower envelope of $L_1$ has slope
$\ge p_{2,i}$. Furthermore we know that at any point in $[Q,2Q]$ all the elements of $L_1$ generate larger or equal cost states at that point using their legacy $p(i,\cdot)$ or $p_{1,i}$ compared to the minimum in $L_1$.

\emph{Case 1 (complete takeover).} Suppose there exists $x$ in $L_1\cap[Q,2Q]$ where an $L_2$
segment attains a cost no higher than the $L_1$ envelope. Since the $L_2$ envelope has slope
$p_{2,i}$ and the $L_1$ envelope has slope $\ge p_{2,i}$ (piecewise-constant), the difference
“$L_2$ minus $L_1$” is nonincreasing in $x$. Hence for all $x'\ge x$ the $L_2$ envelope remains
no higher than the $L_1$ envelope. This implies all $L_1$ segments to the right (and the right
suffix of the segment containing $x$) are dominated.

\emph{Case 2 (no takeover inside).} If for all $x$ in $L_1\cap[Q,2Q]$ the $L_2$ envelope is
strictly higher than the $L_1$ envelope, then $L_2$ is suboptimal there. Outside that band,
$L_1$ may no longer be defined (e.g., beyond $2Q$); the $L_2$ envelope can then become optimal
as the only available candidate, in particular just to the right of $2Q$, up to
$2Q+d(i,i{+}1)$ depending on feasibility. This yields the tail-extension behavior.
\end{proof}

\subsection{MV thresholds and their use}\label{subsec:mv}

Before describing the full procedure, we introduce a threshold quantity that lets us
apply \cref{lem:monotone-legacy-p2,lem:contiguity} together with the dominance tests efficiently.

\paragraph{Setup.}
Let $S_1,S_2$ be two \emph{adjacent} segments in $S(i)$ with explicit states
$\DP(i,f)$ (left) and $\DP(i,g)$ (right), where $g>f$.
Let $c$ denote a unit fuel price available at station $i$ (in our algorithm $c\in\{p_{1,i},p_{2,i}\}$).
For a chosen checkpoint, define $MV(S_1)$ as the \emph{largest} value of $c$ for which the
state generated from $\DP(i,f)$ is no more expensive than the competing state from $S_2$.
Whenever the actual price at $i$ satisfies $c\le MV(S_1)$, $S_1$ wins at that checkpoint; combined
with the dominance lemmas, this lets us remove $S_2$.

\paragraph{Regular MV (boundary checkpoint).}
This applies when, at the boundary $x=g$, \emph{both} segments would buy at station $i$ with
the \emph{same} unit price $c$ (e.g., both use $p_{1,i}$, or both use $p_{2,i}$ after a bulk update).
Equating costs at $x=g$ gives
\[
\DP(i,f)+(g-f)c \ \le\ \DP(i,g)
\quad\Longleftrightarrow\quad
c \ \le\ \frac{\DP(i,g)-\DP(i,f)}{g-f}.
\]
We set
\begin{equation}
\label{eq:mv-regular}
MV(S_1)\;:=\;\frac{\DP(i,g)-\DP(i,f)}{g-f}.
\end{equation}

\paragraph{Terminus MV (used when $S_2$ uses $p_{2,i}$ on its right).}
Let $t=\mathrm{terminus}(S_2)=g+T$ be the last point where $S_2$ uses its ``$p_2$'' fuel
on the right of $g$. In the bulk-updated case of interest here, $S_2$ uses the \emph{current}
discounted price $p_{2,i}$ on $[g,t]$, so
\[
C_2(t)=\DP(i,g)+T\,p_{2,i}.
\]
From $\DP(i,f)$, the additional units to reach $t$ are $L=(g-f)+T$.
Let $a:=min(L,r(i,f))$(If $a$ is equal to $L$ then the formula below can't be used since no amount of $c$ can be used, this case can be resolved when two segments first become adjacent and before any MV value is calculated) be the number of those units that $S_1$ can cover using its
\emph{own legacy} discounted price $p(i,f)$; the remaining $(L-a)$ units must be bought at $i$
at unit price $\min\{c,\,p_{1,i}\}$. 

Two regimes are relevant:

\begin{itemize}
\item \textbf{If $c\ge p(i,f)$}, $S_1$ uses its legacy $p(i,f)$ for $a$ units and price $c$
for the remaining $(L-a)$ units:
\[
C_1(t)=\DP(i,f)+a\,p(i,f)+(L-a)\,c,
\]
so
\begin{equation}
\label{eq:mv-term-regA}
MV_{\mathrm{term}}(S_1)
=\frac{\DP(i,g)+T\,p_{2,i}-\DP(i,f)-a\,p(i,f)}{\,L-a\,},
\qquad\text{(apply with $c\ge p(i,f)$)}.
\end{equation}
When checking against $c=p_{2,i}$, the purchase at $i$ in this regime is $(L-a)$ units, so the
discount $p_{2,i}$ is actually \emph{available} only if $(L-a)\ge Q$; otherwise compare using $c=p_{1,i}$.

\item The formula where no $S_1$ legacy fuel is used must also be calculated. After that the $MV_{\mathrm{term}}$ with the highest value from the two is kept.
\begin{equation}\label{eq:mv-term-regB}
C_1(t)=\DP(i,f)+L\,c
\quad\Longrightarrow\quad
MV_{\mathrm{term}}(S_1)
=\frac{\DP(i,g)+T\,p_{2,i}-\DP(i,f)}{\,L\,}.
\end{equation}
When checking against $c=p_{2,i}$, the discount requires $L\ge Q$ (here all $L$ units are
bought at $i$); if $L<Q$, compare with $c=p_{1,i}$ instead.
\end{itemize}

\paragraph{If $S_2$ uses its \emph{own legacy} $p_2$ on the right.}
In this case $S_2$’s unit price on $[g,t]$ is $p(i,g)$, and we only need to check the \emph{last}
point where that legacy fuel is used, i.e., the same terminus $t$ (not $S_2$’s rightmost coverage).
The formulas above apply verbatim with $p_{2,i}$ replaced by $p(i,g)$ in $C_2(t)$.

\paragraph{Which MV to use.}
Use the \emph{Regular MV} \eqref{eq:mv-regular} at the boundary $x=g$ whenever both segments
would buy at $i$ with the same unit price (equal slopes). Use the \emph{Terminus MV} when the
right segment $S_2$ is using $p_{2,i}$ on its right (bulk-updated case); if $S_2$ is
using its own legacy $p_2$, evaluate the terminus comparison at the last legacy point $t$ as
noted above.

\paragraph{How we store and use $MV$.}
For each adjacent pair $(S_1,S_2)$ we store either $MV_{\mathrm{reg}}(S_1)$ or the appropriate
$MV_{\mathrm{term}}(S_1)$ and maintain these keys in decreasing order. At station $i$, comparing
the actual price $c\in\{p_{1,i},p_{2,i}\}$ against the stored thresholds (and the quantity
conditions $L\ge Q$ or $L-a\ge Q$ when $c=p_{2,i}$) allows us to perform individual updates:
whenever $c\le MV(S_1)$, $S_1$ wins at the checkpoint and $S_2$ is removed with neighbors
re-adjusted.

\paragraph{High-level sketch for constructing the $Q$-approximate $S(i)$ from $S_b(i)$.}
First, determine the state at zero remaining fuel for $S_b(i{+}1)$ by comparing:
(i) the best carry-over state that already reaches $d(i,i{+}1)$ without purchasing at $i$, and
(ii) the cheapest “top-up” among the segments that cannot reach $d(i,i{+}1)$ (adding just enough fuel, possibly $\ge Q$, to reach exactly $d(i,i{+}1)$). Keep the cheaper of the two as $\DP(i{+}1,0)$.

Next, split the current structure at $d(i,i{+}1)$ into:
$X$ = segments with right endpoint $< d(i,i{+}1)$ (cannot reach $i{+}1$) and
$R$ = the remaining segments (can reach $i{+}1$). If a segment straddles $d(i,i{+}1)$, cut it at that point and create a new segment starting at $d(i,i{+}1)$.

Bulk-update $X$ by (lazily) adding $Q$ units at price $p_{2,i}$ to each segment, then run individual updates with $p_{2,i}$ inside $X$ to delete dominated segments. Line-merge the resulting bulk-updated segments with the portion of $R$ that covers the $[Q,2Q]$ band, and then merge $X$ back into the main structure $R$ (applying the relevant lazy updates). Finally, run individual updates with $p_{1,i}$ to remove any remaining dominated segments.

\begin{figure}[ht!]
\centering
\begin{tikzpicture}[
  x=10mm, y=10mm, >=Latex, line cap=round, line join=round,
  every node/.style={font=\small}
]
\def\xmin{-0.5} \def\xmax{10.5}
\def\ybar{1.00} \def\H{0.60} 

\def\aL{0}   \def\aR{3.5} 
\def\bL{2}   \def\bR{6}   
\def\cL{5}   \def\cR{8}   

\def\OLoneL{2}   \def\OLoneR{3.5} 
\def\OLtwoL{5}   \def\OLtwoR{6}   

\tikzset{obox/.style={fill=white, inner sep=1.6pt, outer sep=0pt}}

\draw[->,thick] (\xmin,0) -- (\xmax,0) node[below right] {$x$};
\foreach \t in {0,...,10}{
  \draw (\t, 0.08) -- (\t, -0.08);
  \node[obox,below] at (\t,-0.08) {\t};
}

\draw[thick] (\aL,\ybar) rectangle (\cR,\ybar+\H);

\foreach \x in {\aR,\bL,\bR,\cL}{
  \draw[densely dashed] (\x,\ybar) -- (\x,\ybar+\H);
}

\fill[pattern=north east lines] (\OLoneL,\ybar) rectangle (\OLoneR,\ybar+\H);
\fill[pattern=north east lines] (\OLtwoL,\ybar) rectangle (\OLtwoR,\ybar+\H);

\draw[decorate,decoration={brace,amplitude=5pt}]
  (\aL,\ybar+\H+0.15) -- node[obox,above=4pt] {Segment $A:\,[\aL,\aR]$}
  (\aR,\ybar+\H+0.15);

\draw[decorate,decoration={brace,amplitude=5pt}]
  (\bL,\ybar+\H+0.90) -- node[obox,above=4pt] {Segment $B:\,[\bL,\bR]$}
  (\bR,\ybar+\H+0.90);

\draw[decorate,decoration={brace,amplitude=5pt}]
  (\cL,\ybar+\H+0.15) -- node[obox,above=4pt] {Segment $C:\,[\cL,\cR]$}
  (\cR,\ybar+\H+0.15);

\pgfmathsetmacro{\legendx}{\cR + 0.9}
\pgfmathsetmacro{\legendy}{\ybar + \H - 0.05}

\node[font=\bfseries,anchor=west] at (\legendx,\legendy+0.45) {Legend};

\draw[thick] (\legendx,\legendy+0.10) -- ++(1.0,0);
\node[anchor=west] at (\legendx+1.2,\legendy+0.10) {External boundary (solid)};

\draw[densely dashed] (\legendx,\legendy-0.20) -- ++(1.0,0);
\node[anchor=west] at (\legendx+1.2,\legendy-0.20) {Overlap boundary (dashed)};

\fill[pattern=north east lines] (\legendx,\legendy-0.55) rectangle ++(1.0,0.20);
\node[anchor=west] at (\legendx+1.2,\legendy-0.45) {Overlap area};

\draw[decorate,decoration={brace,mirror,amplitude=5pt}]
  (\legendx,\legendy-1.05) -- ++(1.0,0);
\node[anchor=west] at (\legendx+1.2,\legendy-1.05) {Segment extent (brace)};

\end{tikzpicture}
\caption{The figure represents a number of segments on an X-axis showing their overlapping regions. Notice that the terminus point of a segment is always in the non-overlapping part of a segment while the current endpoints of the segment in the overlap area may be unknown.}
\label{fig:segments_final_no_terminus}
\end{figure}
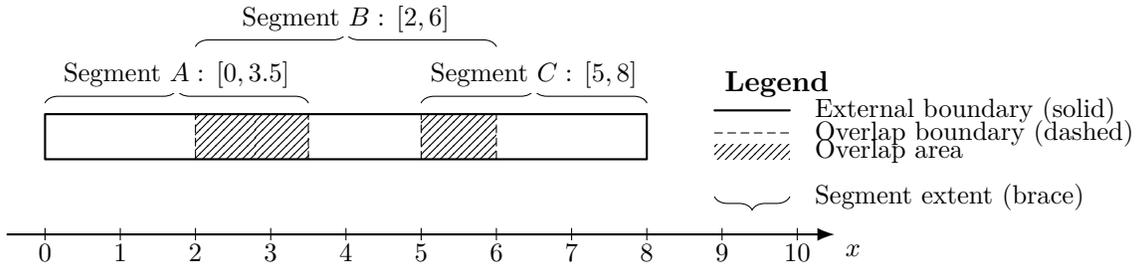

\subsection{Extending the Result to Accommodate Linear Holding Costs}
When holding costs are incorporated into the model, they affect the solution in three ways. First, the value of each state at station $i$ increases by $h(t)$ times its remaining fuel amount (e.g, if it has $x$ remaining fuel it must increase by $xh(t)$ units), reflecting the cost of carrying inventory to the next station. Second, the effective price of $p_{2,i}$ fuel increases by $h(t)$ per unit (see the appendix on how this can be done with lazy propagation). Third, as demonstrated in the following lemma, all MV values increase uniformly by the constant $h(t)$, which can be efficiently implemented as a lazy update operation on the binary search tree maintaining these values.

\begin{lemma}[Holding-cost shift for MV values at station $i$]\label{lem:mv-holding-shift}
Fix station $i$ and a linear holding cost $h_i$ per unit of inventory at the end of period $i$.
Implement the holding-cost update at station $i$ by
\[
\DP(i,x)\ \mapsto\ \DP(i,x)+h_i\,x \quad\text{for all states }x,
\qquad
p_{1,i},\,p_{2,i}\ \mapsto\ p_{1,i}+h_i,\;p_{2,i}+h_i.
\]
Then, for every adjacent pair of segments $(S_1,S_2)$ in $S(i)$ with explicit states
$\DP(i,f)$ and $\DP(i,g)$, $g>f$:
\begin{enumerate}
\item the \emph{Regular MV} (boundary checkpoint $x=g$) satisfies
\[
MV'(S_1)\ =\ MV(S_1)+h_i,
\]
\item and any \emph{Terminus MV} computed at a fixed checkpoint $x_0$
(e.g., $x_0=\mathrm{terminus}(S_2)$) satisfies
\[
MV'_{\mathrm{term}}(S_1)\ =\ MV_{\mathrm{term}}(S_1)+h_i.
\]
\end{enumerate}
Consequently, the ordering of all MV values is preserved by the holding-cost update.
\end{lemma}

\begin{proof}
\emph{Regular MV.} For adjacent states $\DP(i,f)$ and $\DP(i,g)$ with $g>f$, the boundary threshold is
\[
MV(S_1)=\frac{\DP(i,g)-\DP(i,f)}{g-f}.
\]
After the update,
\[
MV'(S_1)
= \frac{\big(\DP(i,g)+h_i g\big)-\big(\DP(i,f)+h_i f\big)}{g-f}
= MV(S_1)+h_i.
\]

\emph{Terminus MV.} Any terminus-based comparison is performed at a fixed checkpoint $x_0$ and yields an inequality
\[
C_1(x_0;c)\ \le\ C_2(x_0),
\]
where $c$ is the unit price used for period-$i$ purchases by $S_1$.
Under the update, both state values acquire the same additive term $h_i x_0$ and all period-$i$ unit prices shift by $+h_i$,
i.e., $c\mapsto c+h_i$. Therefore the new inequality reads
\[
C_1(x_0;c+h_i)+h_i x_0\ \le\ C_2(x_0)+h_i x_0,
\]
which is equivalent to $C_1(x_0;c+h_i)\le C_2(x_0)$.
Thus the new threshold is obtained by adding $h_i$ to the old one, i.e.,
$MV'_{\mathrm{term}}(S_1)=MV_{\mathrm{term}}(S_1)+h_i$.
\end{proof}

\section{The Algorithm}\label{sec:6}
\subsection{General Overview}

Instead of computing individual states, our algorithm computes \emph{segments} that compactly
represent contiguous ranges of optimal states together with the information needed to regenerate
any state they cover. These segments are maintained in an augmented balanced binary search tree
(BST), keyed by their terminus, and augmented with the appropriate $MV$ thresholds (Regular or
Terminus) and lazy tags.

At station $i$, we construct an approximate version of $S(i)$ from $S_b(i)$ using two kinds of updates:

\begin{itemize}
\item \textbf{Individual updates (with $p_{1,i}$ or $p_{2,i}$):} we repeatedly compare the
current price $c\in\{p_{1,i},p_{2,i}\}$ to the \emph{maximum} applicable $MV$ threshold among
adjacent segment pairs (Regular MV when both sides would use the same unit price at the boundary,
Terminus MV when the right segment uses discounted fuel on its right). Whenever $c \le MV$, the
left segment dominates the right one at the checkpoint; we remove the dominated segment and
re-adjust only its neighbors in $O(\log n)$ time.

\item \textbf{Bulk updates (with $Q$ units at $p_{2,i}$):} all segments whose current right
endpoint is strictly less than $d(i,i{+}1)$ (i.e., cannot reach station $i{+}1$) are gathered,
and a lazy ``$+Q$ at $p_{2,i}$'' tag is applied to this subset. If a segment \emph{straddles}
$d(i,i{+}1)$ we cut it at that point (creating at most one additional segment).
\end{itemize}

The per-station flow is:

\begin{enumerate}
\item \textbf{Prune with $p_{1,i}$:} run individual updates comparing $p_{1,i}$ to the stored
$MV$ thresholds to remove any immediately dominated neighbors.

\item \textbf{Form $\mathbf{\DP(i{+}1,0)}$:} compare (i) the best carry-over that already reaches
$d(i,i{+}1)$ (no purchase at $i$) with (ii) the cheapest ``top-up'' among the non-reaching
segments that raises the state to exactly $d(i,i{+}1)$ (if the purchase at $i$ is $<Q$, this
top-up is at $p_{1,i}$; otherwise at $p_{2,i}$). Keep the cheaper as $\DP(i{+}1,0)$.

\item \textbf{Split at $\boldsymbol{d(i,i{+}1)}$:} partition the BST into
$X=\{\text{segments with right end}<d(i,i{+}1)\}$
and $R=\{\text{the rest}\}$; if a segment straddles $d(i,i{+}1)$, cut it.

\item \textbf{Bulk-update $X$ at $\boldsymbol{p_{2,i}}$:} lazily add $Q$ units to each segment
in $X$ (so each can qualify for $p_{2,i}$), then run individual updates \emph{inside} $X$ with
$p_{2,i}$ to remove dominated segments. When both adjacent segments are bulk-updated, we convert
any stored Terminus MV to the Regular MV at their boundary (both sides now use the same unit
price at $g$).

\item \textbf{Line-merge in $\boldsymbol{[Q,2Q]}$:} merge the bulk-updated segments from $X$
with the portion of $R$ that covers the $[Q,2Q]$ band, using the line-merge lemma to discard
suboptimal pieces (complete takeover or tail extension).

\item \textbf{Merge back and final prune:} join the structures, recompute the cross boundary
$MV$ where needed (Regular vs Terminus, per the fuel available on each side), and run a final
individual-update pass with $p_{1,i}$ to eliminate any remaining dominated neighbors.
\end{enumerate}

The BST stores, for each segment, just enough information to recreate any state within its
coverage. Most operations are standard BST primitives (search, split, join, local re-balance)
plus lazy tags for value/position shifts and price updates. Critically, states \emph{between}
two consecutive segments may remain undetermined; if a concrete value is needed at some $x$ between
their termini, we evaluate the two bounding segments at $x$ in $O(1)$ and take the minimum after
an $O(\log n)$ search, yielding overall $O(\log n)$ query time.

\begin{algorithm}[H]
\small
\caption{Per-station construction of $S(i)$ from $S_b(i)$}\label{alg:main}
\DontPrintSemicolon
\KwIn{BST $T$ storing segments of $S_b(i)$ keyed by terminus; prices $p_{1,i},p_{2,i}$; threshold $Q$; distance $d(i,i{+}1)$.}
\KwOut{$T$ updated to a $2Q$-approximate $S(i)$; state $\DP(i{+}1,0)$.}

\For{$i \gets 1$ \KwTo $n$}{
  \tcp*[r]{(a) prune with $p_{1,i}$}
  \textsc{IndividualUpdate}$(T,1)$ \;

  \tcp*[r]{(b) best cost to reach $d(i,i{+}1)$ (carry-over vs top-up)}
  $\mathsf{cand}_{\mathrm{tmp}} \gets \min\Big\{
      \text{carry-over to } d(i,i{+}1),\
      \text{top-up to } d(i,i{+}1)\text{ using } 
      \begin{cases}
         p_{1,i}, & \text{if added }<Q,\\
         p_{2,i}, & \text{if added }\ge Q
      \end{cases}
  \Big\}$ \;

  \tcp*[r]{(c) cut straddler and split at $d(i,i{+}1)$}
  \textsc{CutSegmentAtBoundaryAndPutTheResultInT}\((T, d(i,i{+}1))\)\;
  $(X,R) \gets \textsc{SplitAt}\big(T, d(i,i{+}1)\big)$ \;

  \tcp*[r]{(d) lazy bulk update in $X$ with $Q$ at $p_{2,i}$}
  \textsc{BulkUpdateQAtP2}$(X, Q, p_{2,i})$ \;

  \tcp*[r]{(e) cheapest below-boundary top-up within $X$ (then remove those segments)}
  $\mathsf{cand}_X \gets$ cheapest among segments in $X$ whose terminus $< d(i,i{+}1)$
  after adding just enough $p_{2,i}$ so the segment reaches $d(i,i{+}1)$; remove those segments from $X$ \;

\tcp*[r]{(f) finalize DP(i+1,0) with canonical tie-breaking, then insert/replace at zero}
  \uIf{$\mathsf{cand}_{\mathrm{tmp}} < \mathsf{cand}_X$}{
    $\DP(i{+}1,0) \gets \mathsf{cand}_{\mathrm{tmp}}$\;
  }
  \uElseIf{$\mathsf{cand}_X < \mathsf{cand}_{\mathrm{tmp}}$}{
    $\DP(i{+}1,0) \gets \mathsf{cand}_X$\;
  }
  \Else{
    \tcp{tie: prefer the one whose plan has cheaper legacy p2}
    $\DP(i{+}1,0) \gets \textsc{TieBreakByLegacyP2}(\mathsf{cand}_{\mathrm{tmp}}, \mathsf{cand}_X)$\;
  }
  \textsc{InsertOrReplaceAtZero}$(R, \DP(i{+}1,0))$ \tcp*{CRITICAL: Insert into R, not T}\;

  \tcp*[r]{(g) prune $X$ with $p_{2,i}$}
  \textsc{IndividualUpdate}$(X, 2)$\;

  \tcp*[r]{(h) isolate band, line-merge, and safely reassemble full tree}
  \textsc{CutSegmentAtBoundaryAndPutTheResultInT}$(R, Q)$\;
  \textsc{CutSegmentAtBoundaryAndPutTheResultInT}$(R, 2Q)$\;
  
  \tcp{Partition R to isolate the band without losing data outside of it}
  $(R_{<Q}, R_{\ge Q}) \gets \textsc{SplitAt}(R, Q)$\;
  $(R_{\mathrm{band}}, R_{>2Q}) \gets \textsc{SplitAt}(R_{\ge Q}, 2Q)$\;
  
  $(R'_{\mathrm{band}}, X'_{\mathrm{surviving}}) \gets \textsc{LineMergeBand}(R_{\mathrm{band}}, X)$\;
  
  \tcp{Rejoin all partitions sequentially to restore the global tree T}
  $T_{\mathrm{mid}} \gets \textsc{Join}(R'_{\mathrm{band}}, X'_{\mathrm{surviving}})$\;
  $T_{\mathrm{right}} \gets \textsc{Join}(T_{\mathrm{mid}}, R_{>2Q})$\;
  $T \gets \textsc{Join}(R_{<Q}, T_{\mathrm{right}})$ \;

  \tcp*[r]{(i) cross-boundary dominance (keep Terminus MV at the join)}
  \While{\textsc{CrossBoundaryDominates}$(T)$}{
    \textsc{RemoveDominatedAndRecomputeBoundaryAndMV}$(T)$
  }

  \tcp*[r]{(j) final prune with $p_{1,i}$}
  \textsc{IndividualUpdate}$(T,1)$ \;

  \tcp*[r]{(k) capacity + holding + demand shift}
    \textsc{AddHoldingCostsAndSubtractDemand}$(T,h(i),d(i,i{+}1))$ \;
  \textsc{Trim}$(T,B(i))$ \;
}
\end{algorithm}
\newpage

\begin{lemma}\label{lem:segment-count}
After completing station $i-1$ in Algorithm~\ref{alg:main}—equivalently, at the start of processing station $i$—the maintained $2Q$-approximate boundary solution set $S_b(i)$ contains at most $2i$ segments.
\end{lemma}

\begin{proof}
We prove by induction on $i$.

\textbf{Base ($i=1$).} At the first station there are at most two segments (one using discounted fuel $p_2$ and one using regular fuel $p_1$), so $|S_b(1)|\le 2$.

\textbf{Inductive step.} Assume $|S_b(i)|\le 2i$. We show $|S_b(i{+}1)|\le 2(i{+}1)$.

Consider iteration $i$ of Algorithm~\ref{alg:main}, which transforms $S_b(i)$ into $S_b(i{+}1)$ (after the demand/holding-cost shift in step (k)). The only events that can \emph{increase} the number of segments are:
\begin{enumerate}[label=\textbf{(\Alph*)},leftmargin=*]
\item \emph{A split at the next-station distance.} In step (c), we split the structure at $d(i,i{+}1)$. At most one segment can straddle $d(i,i{+}1)$, so cutting there creates \emph{at most one} extra segment.
\item \emph{Insertion of the zero-fuel state for the next station.} In step (f), we form $\DP(i{+}1,0)$ (implemented as a boundary state at fuel $d(i,i{+}1)$ before the final shift). This may introduce \emph{at most one} new segment (and it may be absorbed immediately if it is dominated).
\end{enumerate}

All other operations cannot increase the count:
\begin{itemize}
\item \emph{Individual updates (with $p_{1,i}$ or $p_{2,i}$)}, the \emph{line-merge in $[Q,2Q]$}, and the \emph{cross-boundary dominance loop} only remove dominated segments or extend/replace them. By the dominance lemmas (terminus dominance and same-slope dominance), when a neighbor is strictly worse at the relevant checkpoint it is dominated on its full coverage, so it is deleted; this never creates new segments.
\item \emph{Bulk update} applies lazily to the ``cannot reach'' subset; subsequent pruning within that subset again only deletes segments. No new internal boundaries are introduced besides the single possible cut at $d(i,i{+}1)$ accounted for in (A).
\item The final capacity/holding/demand operations in step (k) may delete infeasible tails or cut a single straddler at $B(i)$, but these operations do not increase the number of segments.
\end{itemize}

Therefore, from $S_b(i)$ to $S_b(i{+}1)$ the segment count increases by at most $2$ (events (A) and (B)), and may decrease due to pruning. Hence
\[
|S_b(i{+}1)| \le |S_b(i)| + 2 \le 2i + 2 = 2(i{+}1).
\]
This completes the induction.
\end{proof}

{\footnotesize
\begin{algorithm}[H]
\setstretch{0.1}
\caption{\textsc{IndividualUpdate}$(T,j)$}\label{alg:indupdate}
\DontPrintSemicolon
\KwIn{BST $T$; comparator index $j\in\{1,2\}$ with price $p_{j,i}$.}
\KwOut{$T$ with dominated adjacent segments removed under price $p_{j,i}$.}

\While{$\exists$ eligible MV in $T$ for comparator $j$ with value $\ge p_{j,i}$}{
  $(S_1,S_2,g) \gets \textsc{ArgMaxEligibleMV}(T,j)$ \tcp*{adjacent pair at boundary $g$}
  \lIf{\textsc{IsTerminusMV}$(S_1,S_2)$ \textbf{and} \textsc{BoundaryPricesEqualAt}$(g)$}{
    \textsc{RecomputeAsRegularMV}$(S_1,S_2,g)$;\ \textbf{continue}
  }
  \textsc{RemoveDominatedRightNeighborAndRe-adjustLeft}$(S_1,S_2)$ \;
  \textsc{RecomputeLocalMVs}$(\text{left neighbor of }S_1,\,S_1)$ and $(S_1,\,\text{new right neighbor of }S_1)$ \;
}
\textbf{Tie rule:} if $p_{j,i}=MV$ exactly, prefer the left segment at $g$ (canonical choice).\;
\end{algorithm}
\paragraph{Band-merge interval convention.}
In \textsc{LineMergeBand} we treat segment ranges as half-open intervals $[s,e)$.
At the unique switchpoint (where costs are equal), ties are broken in favor of the
bulk-updated family $X$ (i.e., the point belongs to $X$).
\begin{algorithm}[H]
\scriptsize 
\caption{\textsc{LineMergeBand}$(R_{\mathrm{band}}, X_{\mathrm{band}})$ \quad \emph{(Amortized $O(\log n)$ Backward Sweep)}}\label{alg:linemerge}
\DontPrintSemicolon
\KwIn{$R_{\mathrm{band}}$: BST of reaching segments strictly inside $[Q,2Q]$,
      $X_{\mathrm{band}}$: BST of bulk-updated segments strictly inside $[Q,2Q]$.}
\KwOut{$(R',X')$: updated BST fragments after merging.}

\SetKwProg{Fn}{Function}{:}{}
\Fn{\textsc{EvaluateCost}$(\mathit{Tree}, x)$}{
    $\mathit{seg} \gets \textsc{Search}(\mathit{Tree}, x)$\;
    \lIf{$\mathit{seg} \neq \mathrm{null}$}{\Return $\mathrm{Eq}_{\mathit{seg}}(x)$\;}
    \lElse{\Return $\infty$\;}
}
\Fn{\textsc{FindCrossover}$(u, X, s, e)$}{
    \tcp{Binary search down X to find exact algebraic intersection in O(log n)}
    $v \gets \mathrm{root}(X)$;\quad $y^\star \gets e$\; 
    \While{$v \neq \mathrm{null}$}{
        $[s_x, e_x) \gets \text{coverage interval of } v$\;
        $a \gets \max(s_x, s)$;\quad $b \gets \min(e_x, e)$\;
        \If{$a \ge b$}{
            \lIf{$v.\mathrm{terminus} \le s$}{$v \gets v.\mathrm{right}$\;}
            \lElse{$v \gets v.\mathrm{left}$\;}
            \textbf{continue}\;
        }
        
        $\delta_a \gets \mathrm{Eq}_u(a) - \mathrm{Eq}_v(a)$;\quad $\delta_b \gets \mathrm{Eq}_u(b) - \mathrm{Eq}_v(b)$\;
        
        \uIf{$\delta_a \ge 0$}{ $y^\star \gets a$;\quad $v \gets v.\mathrm{left}$\; }
        \uElseIf{$\delta_b < 0$}{ $v \gets v.\mathrm{right}$\; }
        \Else{
            $\alpha_R, \beta_R \gets \text{coeffs of } u$;\quad $\alpha_X, \beta_X \gets \text{coeffs of } v$\;
            \lIf{$\beta_R = \beta_X$}{\Return $a$\;}
            \lElse{\Return $\max\left(a, \min\left(b, \frac{\alpha_X - \alpha_R}{\beta_R - \beta_X}\right)\right)$\;}
        }
    }
    \Return $y^\star$\;
}

\lIf{$R_{\mathrm{band}} = \mathrm{null}$}{\Return $(\mathrm{null}, X_{\mathrm{band}})$\;}
\lIf{$X_{\mathrm{band}} = \mathrm{null}$}{\Return $(R_{\mathrm{band}}, \mathrm{null})$\;}

$y^\star \gets \mathrm{null}$\;
\While{$R_{\mathrm{band}} \neq \mathrm{null}$}{
    $(u, R_{\mathrm{rem}}) \gets \textsc{ExtractMax}(R_{\mathrm{band}})$\;
    
    $s \gets Q$\;
    \If{$R_{\mathrm{rem}} \neq \mathrm{null}$}{
        $\mathit{prev} \gets \textsc{FindMax}(R_{\mathrm{rem}})$\;
        $s \gets \max(Q, \mathit{prev}.\mathrm{terminus})$\;
    }
    $e \gets \min(2Q, u.\mathrm{terminus})$\;
    
    $\Delta_s \gets \mathrm{Eq}_u(s) - \textsc{EvaluateCost}(X_{\mathrm{band}}, s)$\;

    \uIf{$\Delta_s \ge 0$}{
        \tcp{X is cheaper at s. By monotonicity, X strictly dominates u. u is deleted.}
        $y^\star \gets s$;\quad $R_{\mathrm{band}} \gets R_{\mathrm{rem}}$\;
    }
    \Else{
        $\Delta_e \gets \mathrm{Eq}_u(e) - \textsc{EvaluateCost}(X_{\mathrm{band}}, e)$\;
        
        \uIf{$\Delta_e \le 0$}{
            \tcp{R dominates X left of e. Restore u and terminate sweep.}
            $y^\star \gets e$;\quad $R_{\mathrm{band}} \gets \textsc{Join}(R_{\mathrm{rem}}, \textsc{MakeSingleton}(u))$\;
            \textbf{break}\;
        }
        \Else{
            $y^\star \gets \textsc{FindCrossover}(u, X_{\mathrm{band}}, s, e)$\;
            $R_{\mathrm{band}} \gets \textsc{Join}(R_{\mathrm{rem}}, \textsc{MakeSingleton}(u))$\;
            \textbf{break}\;
        }
    }
}
\uIf{$y^\star \neq \mathrm{null}$}{
    \textsc{CutSegmentAtBoundary}$(R_{\mathrm{band}}, y^\star)$\;
    \textsc{CutSegmentAtBoundary}$(X_{\mathrm{band}}, y^\star)$\;

    $(R_{<}, R_{\ge}) \gets \textsc{SplitAt}(R_{\mathrm{band}}, y^\star)$\;
    $(X_{<}, X_{\ge}) \gets \textsc{SplitAt}(X_{\mathrm{band}}, y^\star)$\;

    \textsc{DeleteTree}$(R_{\ge})$;\quad \textsc{DeleteTree}$(X_{<})$\;
    \Return $(R_{<}, X_{\ge})$\;
}
\Else{
    \Return $(\mathrm{null}, X_{\mathrm{band}})$\;
}
\end{algorithm}

\subsection{Auxiliary procedures used by Algorithm~\ref{alg:main}}
We summarize the auxiliary procedures used by Algorithm~\ref{alg:main}. Each
procedure operates on the height-balanced BST of segments augmented with MV keys
(Regular or Terminus) and standard lazy tags; asymptotic costs assume $O(\log n)$
BST edits and $O(\log n)$ priority updates.

\begin{description}
  \item[\textsc{SplitAt}$(T,x)$:] Split the segment tree at the boundary
  $x=d(i,i{+}1)$ into $X$ (right endpoint $<x$) and $R$ (the rest). Time: $O(\log n)$.

  \item[\textsc{CutSegmentAtBoundaryAndPutTheResultInT}$(T,d(i,i{+}1))$:] If some segment spans $x=d(i,i{+}1)$,
  cut it there and recompute the two new boundary MVs. Time: $O(\log n)$.

  \item[\textsc{BulkUpdateQAtP2}$(X,Q,p_{2,i})$:] Lazily enable a $Q$-block at station $i$
  priced at $p_{2,i}$ for all segments in $X$. Time: $O(1)$ (lazy tag).

  \item[\textsc{IndividualUpdate}$(T,j)$:] Repeatedly remove dominated right neighbors under
  the comparator price $p_{j,i}$ ($j\in\{1,2\}$), using the maximum eligible MV at each step;
  converts Terminus$\to$Regular at equal boundary prices. Time: $O(\log n)$ per removal;
  $O(n\log n)$ total over the run.

  \item[\textsc{InsertOrReplaceAtZero}$(R,\DP(i{+}1,0))$:] Insert (or replace) the
  boundary state/segment at $x=d(i,i{+}1)$ in the \emph{reaching} tree $R$, with canonical tie rule
  (cheaper legacy $p_2$, then fixed key). Time: $O(\log n)$.

  \item[\textsc{LineMergeBand}$(R\!\mid_{[Q,2Q]},X)$:] Merge bulk-updated $X$ against the
  slice of $R$ covering $[Q,2Q]$ using \cref{lem:linemerge-band}; returns trimmed fragments to rejoin.
  Time: $O(\#\mathrm{edits}\cdot\log n)$, with $\#\mathrm{edits}=O(n)$ over the run.

  \item[\textsc{CrossBoundaryDominates}$(T)$:] At the unique $R$--$X$ join, test the correct
  MV (Terminus or Regular) against the appropriate comparator price; return true if the right
  neighbor is dominated. Time: $O(1)$.

  \item[\textsc{RemoveDominatedAndRecomputeBoundaryMV}$(T)$:] Remove the dominated neighbor
  at the join; recompute the new cross-boundary MV (Regular if boundary prices are equal,
  else Terminus). Time: $O(\log n)$.

  \item[\textsc{Join}$(R',X')$:] Join two trees with disjoint key ranges; recompute the
  unique cross-boundary MV created by the join. Time: $O(\log n)$.

  \item[\textsc{Trim}$(T,B(i))$:] Enforce capacity by removing infeasible tails and cutting
  at $B(i)$. Time: $O(\log n)$.

  \item[\textsc{AddHoldingCosts}$(T,h(i))$:] Apply end-of-period linear holding costs by a
  lazy value shift $+h(i)\cdot x$ and a uniform MV-key shift $+h(i)$
  (holding-cost MV-shift lemma). Time: $O(1)$ (lazy).
\end{description}

\subsection{Correctness of the Algorithm}\label{sec:correctness}

We prove that Algorithm~\ref{alg:main} produces, at every station $i$,
a correct $2Q$-approximate solution set $S(i)$ and the correct boundary state
$\DP(i{+}1,0)$. The proof proceeds by induction over stations and is
structured around an invariant that is preserved by each step of the algorithm.

\paragraph{Invariant $(\mathcal{I}_i)$ (coverage, optimality, and consistency).}
After completing the loop for station $i$, the BST $T$ encodes a $2Q$-approximate
solution set $S(i)$ with the following properties:
\begin{itemize}
  \item[(I1)] \textbf{Coverage and optimality on $[0,2Q]$.}
    For every inventory level $x\in[0,2Q]$, the value represented implicitly or explicitly by $T$ equals
    the true optimum $\DP(i,x)$. Each segment stores one explicit state
    and a valid linear equation to generate all covered states, and (by the
    continuity lemma in Section~\ref{sec:5}) every point covered by a segment is
    generated by its anchor without gaps for the chosen fuel type. In the segment overlapping sections the optimal value of a state is calculated by taking the appropriate minimum of the states generated by the two bounding segments.
  \item[(I2)] \textbf{No removable domination at the maintained price.}
    No adjacent pair in $T$ violates the dominance tests under $p_{1,i}$
    (Regular or Terminus, as appropriate), i.e., the left segment is not strictly
    better at the boundary checkpoint when evaluated with the correct unit price;
    equivalently, all eligible MVs are $<p_{1,i}$ after the final prune.
  \item[(I3)] \textbf{Correct MV type and value at every boundary.}
    If both sides buy at the same unit price at the boundary, the stored MV is
    Regular; otherwise it is Terminus and computed at the right segment's
    terminus checkpoint, per Section~\ref{subsec:mv}.%
\end{itemize}

\paragraph{Base case ($i=1$).}
At station $1$ there is no legacy $p_2$; the algorithm constructs the reachable
states using $p_{1,1}$ and (when applicable) $p_{2,1}$, applies
\textsc{IndividualUpdate}$(T,1)$, and forms the boundary state. The continuity
lemma guarantees that each segment covers a single interval; the dominance and
MV rules ensure (I2)–(I3). Thus $(\mathcal{I}_1)$ holds.

\paragraph{Inductive step.}
Assume $(\mathcal{I}_i)$ holds. We show that, after executing steps (a)–(k) of
Algorithm~\ref{alg:main}, $(\mathcal{I}_i)$ still holds for the constructed
$S(i)$ and the boundary state $\DP(i{+}1,0)$ is optimal. This yields
$(\mathcal{I}_{i+1})$ in the next iteration once the demand $d(i,i{+}1)$ is
subtracted (which is exactly what the boundary split/merge achieves).

\medskip
\noindent\textbf{Step (a): \textsc{IndividualUpdate}$(T,1)$.}
Each iteration compares $p_{1,i}$ to the maximum eligible MV. If the MV is
Terminus but both sides now buy at the same unit price at the boundary, the MV
is converted to Regular (no structural change). Otherwise, when $p_{1,i}$ is not
larger than the MV, the left segment is strictly cheaper at the checkpoint, and
by the dominance-from-terminus test (Section~\ref{sec:5}) it dominates the right
segment on its whole coverage; removing the dominated segment preserves (I1).
The loop ends exactly when all eligible MVs are $<p_{1,i}$, establishing (I2).

\medskip
\noindent\textbf{Step (b): forming the boundary candidate $\mathsf{cand}_{\mathrm{tmp}}$.}
By the generalized Observation~1, a state with remaining fuel $x>d(i,i{+}1)$
must not buy at station $i$, because the same or cheaper fuel is available at
station $i{+}1$ in the non-increasing price regime; hence the only optimal ways
to be exactly at $d(i,i{+}1)$ are: (i) carry over a state that already reaches
the boundary, or (ii) top up a non-reaching state at $i$. By the all-units
discount, the top-up is either $<Q$ (priced at $p_{1,i}$) or $\ge Q$
(discounted, priced at $p_{2,i}$). Taking the minimum of these two possibilities
therefore yields the optimal boundary candidate among the reaching side $R$,
proving that $\mathsf{cand}_{\mathrm{tmp}}$ is correct.

\medskip
\noindent\textbf{Step (c): split at $d(i,i{+}1)$ and cut straddlers.}
These are structural edits that do not change any state value. The explicit
boundary guarantees that subsequent tests compare exactly the intended
checkpoints (either $x=g$ for Regular MVs or $x=t$ for Terminus MVs).

\medskip
\noindent\textbf{Step (d): bulk-update $X$ at $p_{2,i}$.}
Only segments with right endpoint $<d(i,i{+}1)$ (i.e., cannot reach $i{+}1$)
are allowed to buy at station $i$, by Observation~\ref{obs:no-refuel-if-reaches-next}. Bulk-enabling a $Q$-block
at $p_{2,i}$ for all of them is logically equivalent to exposing the cheapest
available station-$i$ unit price they can obtain when purchasing $\ge Q$.
No state in $R$ is altered at this step.

\medskip
\noindent\textbf{Step (e): still-below-boundary top-ups in $X$.}
Among segments in $X$ whose terminus remains $<d(i,i{+}1)$ even after the bulk update in (d), the only way to reach the boundary is to add the exact shortfall priced at $p_{2,i}$ (the $Q$-block having been enabled). We evaluate all such candidates, keep the cheapest one at $d(i,i{+}1)$ as $\mathsf{cand}_X$, and delete the rest. This deletion is perfectly safe due to the monotonicity of parallel extensions. Because all these segments must use the exact same unit price ($p_{2,i}$) to extend to $d(i,i{+}1)$ and beyond it, their cost functions for $x \ge d(i,i{+}1)$ are strictly parallel affine lines. Among parallel lines, the segment that attains the strict minimum at $d(i,i{+}1)$ will therefore strictly dominate all other such segments for every point $x > d(i,i{+}1)$. By compressing this winning envelope into a single boundary state carrying the legacy price $p_{2,i}$, we perfectly preserve the optimal frontier while safely discarding redundant segments.
\medskip
\noindent\textbf{Step (f): finalize $\DP(i{+}1,0)$ and insert into $R$.}
By the argument in steps (b) and (e), the optimal boundary state is the minimum
of $\mathsf{cand}_{\mathrm{tmp}}$ and $\mathsf{cand}_X$; the tie rule (prefer the one with the
cheaper $p_2$, then a fixed canonical rule) preserves optimality and
ensures determinism. Inserting this single boundary point into $R$ maintains
(I1) on the reaching side.

\medskip
\noindent\textbf{Step (g): \textsc{IndividualUpdate}$(X,2)$.}
Now that purchases in $X$ can be made at $p_{2,i}$, the same dominance logic as
in step (a) applies with comparator $p_{2,i}$. Using the Regular or Terminus MV
as appropriate, any right neighbor strictly dominated at the tested price is
removed; by the dominance lemma, no optimal state is lost and (I1) is preserved
inside $X$.

\medskip
\noindent\textbf{Step (h): line-merge in $[Q,2Q]$.}
By \cref{lem:linemerge-band}, merging the bulk-updated segments against $R$'s $[Q,2Q]$ slice yields exactly two patterns:
complete takeover (once $X$ is better at some point in the band, it remains
better to the right because it uses a weakly cheaper unit price), or tail
extension (if $X$ never wins inside the band, it can only become optimal to the
right of $2Q$). The line-merge performs exactly these trims, preserving (I1).

\medskip
\noindent\textbf{Step (i): cross-boundary dominance at the $R$--$X$ join.}
At the unique join, the correct MV type is used: Terminus if the right segment
still uses legacy $p_2$ up to its terminus, else Regular at the boundary.
If the comparator price meets/exceeds the threshold, the right neighbor is
dominated by the left and can be removed (dominance lemma). The loop stops
exactly when the new boundary is non-dominated, establishing (I2) at the join.

\medskip
\noindent\textbf{Step (j): final \textsc{IndividualUpdate}$(T,1)$.}
This is identical to step (a) on the reassembled structure, ensuring that no
removable domination under $p_{1,i}$ remains anywhere in $T$ (I2), while (I1)
is preserved by the dominance lemma.

\medskip
\noindent\textbf{Capacity and holding costs.}
\textsc{Trim} removes only infeasible states (those exceeding $B(i)$) and cuts a
straddler at $B(i)$; feasibility is preserved and (I1) remains true on the
allowed domain. \textsc{AddHoldingCostsAndSubtractDemand} adds $h(i)\cdot x$ lazily to every state and subtracts the appropriate demand.
By \cref{lem:mv-holding-shift}, every MV threshold
increases by the same constant $h(i)$, so the ordering of thresholds and all
dominance decisions are unchanged; hence (I2)–(I3) remain true.

\medskip
\noindent\textbf{Conclusion of the inductive step.}
Combining the above, every step either (i) performs a purely structural edit,
(ii) removes a segment that is dominated everywhere on its coverage (hence cannot
be part of an optimal solution), or (iii) inserts the unique optimal boundary
state $\DP(i{+}1,0)$. Therefore $(\mathcal{I}_i)$ holds for the
constructed $S(i)$, and the boundary state is optimal. As $S_b(i{+}1)$ is
obtained from $S(i)$ by subtracting $d(i,i{+}1)$ from the remaining fuel, the same
arguments yield $(\mathcal{I}_{i+1})$.

\begin{theorem}\label{thm:correctness}
Assuming non-increasing unit prices over time, a single all-units discount
threshold $Q$, and linear holding costs, Algorithm~\ref{alg:main} correctly
computes, for every station $i$, a $2Q$-approximate solution set $S(i)$ and the
optimal boundary state $\DP(i{+}1,0)$. Equivalently, for every
$x\in[0,2Q]$, the value represented in $T$ equals $\DP(i,x)$, and the
state $\DP(i{+}1,0)$ is optimal.
\end{theorem}

\subsection{Time Complexity Analysis}\label{sec:complexity}

\paragraph{Model of computation.}
We store segments in a height-balanced BST keyed by their terminus. Standard
operations (search, insert, delete, split, join) take $O(\log n)$ time. The
data structure maintains at most $O(n)$ MV keys (Regular or Terminus); reading or updating takes $O(\log n)$ time.
Bulk and holding-cost updates are implemented lazily (constant-time tags
applied at subtree roots or global MV-key shifts).

\paragraph{Global structural bound.}
Let $m_i$ be the number of segments in $S(i)$. By \cref{lem:segment-count}, at each station $i$ we create \emph{at most} two new
segments (a split at $d(i,i{+}1)$ and the explicit zero-inventory state for
station $i{+}1$); every dominated segment is removed permanently. Hence
$m_i\le 2i$ and the total number of segment insertions and deletions across
the entire run is $O(n)$. The same $O(n)$ bound holds for the number of MV
keys that are ever created and subsequently removed.

\paragraph{Per-station operations accounting.}
For a fixed station $i$, we analyze each step of Algorithm~\ref{alg:main}.

\begin{itemize}
\item \textbf{(a) \textsc{IndividualUpdate}$(T,1)$.}
Each loop iteration removes one dominated segment or converts one Terminus MV
to a Regular MV and continues. Across \emph{all} stations, there are at most
$O(n)$ removals (each segment is removed once) and at most $O(n)$ conversions of MV values.
Each action costs $O(\log n)$
to update the BST, so the total over the run is $O(n\log n)$; per
station this contributes $O(\log n)$ amortized.

\item \textbf{(b) Boundary candidate $\mathsf{cand}_{\mathrm{tmp}}$.}
A search at $x=d(i,i{+}1)$ returns the at-most-two bounding segments; cost
evaluations are $O(1)$. Total $O(\log n)$ per station.

\item \textbf{(c) \textsc{CutSegmentAtBoundaryAndPutTheResultInT} + \textsc{SplitAt}.}
Cutting a single segment (if any) and splitting the tree by the boundary
are both $O(\log n)$.

\item \textbf{(d) \textsc{BulkUpdateQAtP2}$(X,\cdot)$.}
Lazy annotation on the subtree root is $O(1)$; pushing cost is charged to
subsequent searches/updates.

\item \noindent\textbf{Step (e): still-below-boundary top-ups in $X$.}
Among segments in $X$ whose terminus remains $<d(i,i{+}1)$ even after the bulk update in (d), the only way to reach the boundary is to add the exact shortfall priced at $p_{2,i}$ (the $Q$-block having been enabled). We evaluate all such candidates, keep the cheapest one at $d(i,i{+}1)$ as $\mathsf{cand}_X$, and delete the rest. This deletion is perfectly safe due to the monotonicity of parallel extensions. Because all these segments must use the exact same unit price ($p_{2,i}$) to extend to $d(i,i{+}1)$ and beyond it, their cost functions for $x \ge d(i,i{+}1)$ are strictly parallel affine lines. Among parallel lines, the segment that attains the strict minimum at $d(i,i{+}1)$ will therefore strictly dominate all other such segments for every point $x > d(i,i{+}1)$. By compressing this winning envelope into a single boundary state carrying the legacy price $p_{2,i}$, we perfectly preserve the optimal frontier while safely discarding redundant segments.

\item \textbf{(f) Insert $\DP(i{+}1,0)$ into $R$.}
One insertion/replacement at the boundary in the tree $R$ costs
$O(\log n)$.

\item \textbf{(g) \textsc{IndividualUpdate}$(X,2)$.}
After the bulk update, internal adjacencies in $X$ are eligible for pruning
at $p_{2,i}$. Exactly as in step (a), over the whole run there are $O(n)$
removals/conversions; total $O(n\log n)$ and $O(\log n)$ amortized per
station.

\item \textbf{(h) \textsc{LineMergeBand}$(R\!\mid_{[Q,2Q]},X)$.}
The merge procedure sweeps backwards (right-to-left) through the segments of $R$. Because the cost difference $\Delta(x)$ is monotonic (Lemma \ref{lem:linemerge-band}), any segment $u \in R$ where $X$ is cheaper at the left boundary is entirely dominated by $X$. We permanently delete $u$ via an $O(\log n)$ tree update. The sweep terminates immediately upon finding the single crossover point or a segment where $R$ is optimal. Because every non-terminating loop iteration permanently deletes a segment, the total traversal work across all merges is bounded by the $O(n)$ maximum segments created, yielding an amortized $O(\log n)$ cost per station.

\item \textbf{(i) Cross-boundary dominance loop.}
At the unique $R$--$X$ join, the loop repeatedly removes a dominated neighbor
and recomputes the new cross-boundary MV (Regular if boundary prices become
equal, else Terminus). As each iteration deletes a segment, there can be at
most $O(n)$ iterations in total; each iteration is $O(\log n)$. Thus
$O(n\log n)$ overall and $O(\log n)$ amortized per station.

\item \textbf{(j) Final \textsc{IndividualUpdate}$(T,1)$.}
Same accounting as (a): $O(n\log n)$ total, $O(\log n)$ amortized.

\item \textbf{Capacity/holding updates.}
\textsc{Trim}$(T,B(i))$ performs at most one cut (and possibly deletes a tail
of segments); across the run these deletions are covered by the global $O(n)$
removal bound, so $O(\log n)$ amortized per station. \textsc{AddHoldingCostsAndSubtractDemand}
$(T,h(i),d(i,i{+}1))$ is a lazy value/MV-key shift and is $O(1)$ per station.
\end{itemize}

\paragraph{Putting it together.}
Every station performs a constant number of $O(\log n)$ structural operations
(split/join/insert) and a constant number of dominance passes whose
\emph{total} number of removals/conversions across the run is $O(n)$. Since
each removal/conversion/re-adjustment costs $O(\log n)$, the total cost over $n$ stations
is $O(n\log n)$.

\begin{theorem}\label{thm:complexity}
Under the segment bound $m_i\le 2i$ and with augmented balanced-tree primitives, the
per-station construction of $S(i)$ from $S_b(i)$ runs in total time
$O(n\log n)$ and uses $O(n)$ space.
\end{theorem}

\subsection{An extension for $B(i)>2Q$}
This add-on is used \emph{only} to answer $\DP(i,x)$ for $x\le B(i)$ using purchases
at station $i$; it does \emph{not} modify how we build $S_b(i{+}1)$. It is executed
\emph{after} the per-station routine has finished (i.e., after the final prune with $p_{1,i}$)
and is fully reversed before processing station $i{+}1$.

\paragraph{Transient in-place expansion (no line-merge).}
Let each surviving segment $S$ have anchor $\DP(i,f)$, legacy discounted price $p(i,f)$
with residual $r(i,f)$, and terminus $t_S:=f+r(i,f)$. Proceed as follows:
\begin{enumerate}
  \item \textbf{Select eligible segments.}
        Mark the set
        \[
          X^{\mathrm{cap}} \ :=\ \{\, S\in T \;:\; t_S + Q \;<\; B(i) \,\},
        \]
        i.e., exactly those segments that can legally accept a full $Q$-block at price $p_{2,i}$
        without exceeding capacity.
  \item \textbf{Apply a lazy bulk tag (in place).}
        On $X^{\mathrm{cap}}$ only, apply a transient lazy tag that grants a $Q$-block at
        unit price $p_{2,i}$ (all-units discount at station $i$). No line-merge is performed.
  \item \textbf{Local pruning inside $X^{\mathrm{cap}}$ and global reconnection.}
        Run \textsc{IndividualUpdate} limited to \emph{internal} adjacencies of
        $X^{\mathrm{cap}}$ with comparator $p_{2,i}$. This improves query speed but must not
        remove or alter segments outside $X^{\mathrm{cap}}$ and must not use cross-boundary
        comparisons; the base frontier remains intact.
  \item \textbf{Answer queries at station $i$.}
        To evaluate $\DP(i,x)$ for any $x\le B(i)$, perform the usual $O(\log n)$ search
        for the at-most-two bounding segments and compute their costs at $x$ \emph{with} the
        lazy tag in effect for any segment in $X^{\mathrm{cap}}$. Take the minimum.
  \item \textbf{Reverse before $i{+}1$.}
        Remove the lazy tag from $X^{\mathrm{cap}}$ and discard any per-query caches so that
        $T$ returns exactly to its post-step-(j) state. This prepares the tree for station
        $i{+}1$ and ensures the add-on is never used to construct $S_b(i{+}1)$.
\end{enumerate}

\paragraph{Correctness and cost.}
The transient tag is never used to form $S_b(i{+}1)$, so Observation~\ref{obs:no-refuel-if-reaches-next} is respected.
All comparisons during queries are pointwise and use the same piecewise-affine equations as
in the core method; non-increasing prices guarantee that enabling $p_{2,i}$ is consistent
with the dominance logic used inside $X^{\mathrm{cap}}$. The add-on does not change the global
$O(n\log n)$ bound; tagging and untagging are $O(1)$, and each query remains $O(\log n)$.

\subsection{Keeping the optimal decisions}
Keeping the optimal decisions can be done in two ways:
\begin{enumerate}
  \item \textbf{Per-station operation log.} At each station, record every operation on the tree (individual updates, bulk updates) together with information about the \emph{affected} elements. With this record, one can retrace every step for a given \emph{solution} by \emph{reconstructing the tree backward}. This takes $O(n\log n)$ time.
  \item \textbf{Lazy branch tracking.} A more sophisticated approach lazily records, for each branch of the augmented tree $T$, the sequence of operations applied to that branch. This requires more careful pointer management to preserve operation order. When a sequence is pushed from a node to its children, keep in each child’s list a pointer to a compact single-element proxy that references the original sequence; this lets the children’s lists \emph{share} the parent’s sequence without copying.
\end{enumerate}

\section{Conclusion}
We presented an $O(n \log n)$ algorithm for the single-item lot sizing problem with a 1-breakpoint all-units quantity discount in the case of non-increasing prices.
Our algorithm works identically for integer and non-integer quantities.

Our algorithm is an improvement over the previous $O(n^2)$ state-of-art algorithm for the same problem.
Furthermore, we developed a rather general dynamic programming technique that may be applicable to similar problems in lot sizing or elsewhere.
Finally an interesting avenue of research would be to investigate broad classes of cost and holding functions that preserve the validity of the lemmas of our algorithm.

\bibliographystyle{plain}
\bibliography{references}

\section{Appendix}
We will first define the operations required by our algorithm and then explain how they can be used to construct the functions in the pseudocode provided in the previous section and provide the time complexity for each operation.
For these operations we use an augmented BST, where each of its internal nodes keeps additional information about the maximum MV values of their children.

A state segment is a 7-tuple containing the leftmost state (position and value) of the segment, its terminus position, its MV and its type (regular or terminus), its last used $p_2$ and the additional amount of that fuel it can use after its terminus point.
For deletions and insertions, what matters for a state is the terminus point of a segment and the MV value; the rest of the information is just carried with it.
The data structure needs for every operation to also apply each lazy tag on the visited nodes and then push it downward to its children.
A data structure that supports the following basic operations in $O(\log n)$ time is needed in order to implement the functions of our main algorithm:

\begin{description}
    \item[Insert$(x,w,v)$:] Inserts a segment (whose rightmost position/terminus is $x$) at position $x$ along a one-dimensional axis. The element is given an MV of $w$ and a value $v$. Its position is its terminus point.
    \item[Remove$(e)$:] Given a pointer $e$, it deletes the element $e$ from the structure.

    \item[Search$(x)$:] Identifies the optimal segment $t$ that covers the point $x$. The operation returns a reference to all the segments (2 at most) that overlap with $x$ on the one-dimensional axis. This is done by searching for the element with the lowest terminus higher than $x$ and the element with the highest terminus less than $x$. We can find the optimal segment by generating the state using each of the two bounding segments and taking the minimum.
    \item[ChangeMV$(w,e)$:] Changes the MV to $w$ of a specified segment $e$, identified by a reference $e$ to the element.
    \item[FindMax:] Locates the segment with the highest MV among all elements in the structure.
    \item[RemoveMax:] Deletes the segment with the highest MV from the structure and reconfigures the position and MV of its adjacent segments. Subsequently it must update the maximum MV at each node in the path from that segment to the root starting from the lowest node.
    \item[IncreaseMV$(x,d)$:] Increases the MV of all elements from the element with the lowest position in the structure up to, but not including, the one that starts at or after position $x$. Each of these elements receives an increment of $d$ units in MV.
    \item[IncreasePosition$(x,d)$:] Increases the position by $d$ of all elements starting from the element with the lowest position in the structure up to, but not including, the one that starts at or after position $x$. If the segment starts before $x$ but ends after it we cut it at that position. Each of these elements receives an increment of $d$ units in position value.
    \item[IncreaseValue$(x,d)$:] Increases the value by $d$ of all segments from the segment with the lowest position in the structure up to, but not including, the one that starts at or after position $x$. Each of these elements receives an increment of $d$ units in value.
    \item[ChangePrice$(x,a)$:] Changes the price of $p_2$ by $a$ of all segments from the segment with the lowest position in the structure up to, but not including, the one that starts at or after position $x$. Each of these elements has the price of its cheapest fuel set to $a$.

    \item[Split$(x)$:] Cuts the tree into two trees $L$ and $R$, with $L$ containing all the elements with positions lower or equal to $x$ and $R$ containing the rest of the elements.
    \item[Join$(L,R)$:] Merges the trees $L$ and $R$ into a balanced binary search tree given that all the elements of one of the trees are at higher positions than the elements of the other tree.
\end{description}

\subsection{Data Structure Implementation}
The operations described above are implemented using an \textbf{augmented balanced binary search tree} (such as an AVL tree or Red-Black tree) where nodes are ordered by their terminus positions.

\subsubsection{Position Coordinate System}
We use an absolute coordinate system for inventory positions:
\begin{itemize}
\item Position 0 corresponds to zero inventory at the start of station 1
\item For any station $i$, the position at zero inventory is $\sum_{j=1}^{i-1} d_j$ (cumulative distance)
\item We maintain a prefix sum array: $D[i] = d(1,i-1)= \sum_{j=1}^{i-1} d_j$ for efficient position calculations
\item A state with inventory $f$ at station $i$ has absolute position $D[i] + f$
\item We can easily convert the absolute position system to a relative one which starts with zero at any station $i$ by subtracting from all states the cumulative sum of distances up to that station.
\end{itemize}

This absolute positioning system allows us to:
\begin{itemize}
\item Compare positions across different stations uniformly
\item Efficiently determine which segments can reach the next period
\item Apply bulk updates by shifting absolute positions
\end{itemize}

\subsubsection{Node Structure}
Each node in the BST stores:
\begin{enumerate}
\item \textbf{Primary data:} The complete 7-tuple segment representation
\item \textbf{Auxiliary data for efficiency:}
\begin{itemize}
\item \texttt{max\_MV}: Maximum MV value in this node's subtree
\item \texttt{max\_MV\_node}: Pointer to the node with maximum MV in subtree
\item \texttt{has\_terminus\_MV}: Boolean indicating if subtree contains terminus-type MVs
\end{itemize}
\item \textbf{Lazy propagation fields:}
\begin{itemize}
\item \texttt{position\_delta}: Pending position increase for entire subtree
\item \texttt{value\_delta}: Pending $\DP$ value increase for entire subtree
\item \texttt{MV\_delta}: Pending MV increase for entire subtree (for holding costs)
\item \texttt{p2\_price}: Pending $p_2$ fuel price accessible for the segments of the entire subtree
\end{itemize}
\end{enumerate}

\subsubsection{Derived Attributes}
Some segment attributes are computed dynamically as needed rather than stored:
\begin{itemize}
\item Once the latest $p_2$ price is pushed to a segment then the actual $p(i,f)$ price is changed and stored in the segment tuple.
Some attributes of a segment can be derived from the index of the station where it last received a bulk-update.
\item The remaining capacity $r(i,f)$ is computed as:
\[
r(i,f) = \min \left( \text{segment.r\_fuel},\, \min_{i' \le k \le i} \{ B(k) - d(k, i) \} - f \right)
\]
where $i'$ is the period where the segment last received a bulk update, and $f$ is the remaining fuel of the explicit state at station $i$.
\item The current optimal fuel price for a segment considers both its stored $p_2$ price and the current period's $p_{1,i}$
\end{itemize}

\subsubsection{Maintaining Auxiliary Information}
After each modification (insertion, deletion, or lazy update), we update auxiliary fields bottom-up:
\begin{itemize}
\item \texttt{max\_MV} = $\max$(node.MV, left.max\_MV, right.max\_MV)
\item \texttt{max\_MV\_node} points to the node achieving this maximum
\item \texttt{has\_terminus\_MV} = node.is\_terminus
\end{itemize}

This augmentation enables $O(\log n)$ access to the maximum MV and $O(\log n)$ updates while maintaining the tree's balance properties.

Since these lazy updates always affect the entirety of the affected tree (either because they are applied to a subtree cut from the main tree or because they are applied as a holding cost across the whole tree), they are always inserted at the root. For a discussion of a more general version of lazy updates, see \textit{Purely Functional Data Structures} by Chris Okasaki.

The operations split and join are by far the hardest from the operations of the augmented tree structure. These operations are designed to accommodate bulk updates. The states affected by a bulk update are first split into a separate tree on which the update is applied lazily. After the update, the resulting tree is reconnected to the remainder of the tree, and the MV value between the last non-bulk-updated element and the first bulk-updated element is recalculated.
We include a slightly modified version of a standard join operation of balanced trees below:
\subsection{Split and Join in $O(\log n)$ time (AVL, parent pointers, strict lazy)}\label{app:split-join}

The operations \texttt{Split} and \texttt{Join} are the key structural edits used by the main algorithm to isolate the ``cannot-reach'' prefix, apply lazy bulk tags to it, and then reconnect it. For the overall $O(n\log n)$ time bound to hold, we need \texttt{Split} and \texttt{Join} to run in strict $O(\log n)$ time while preserving the AVL balance invariants.

Standard AVL rotations only shift heights by exactly 1, making them incapable of resolving the massive height discrepancies generated when a tree is split. Therefore, our \texttt{Split} algorithm avoids local rotations entirely. Instead, it systematically dismantles the tree along the search path and reconstructs it using \texttt{JoinWithPivot}, which dynamically bridges arbitrary height differences by climbing down the spine of the taller subtree.

\paragraph{Conventions and invariants.}
Each AVL node stores: \texttt{key} (segment terminus), pointers \texttt{left}, \texttt{right}, \texttt{parent}, cached \texttt{height}, all augmented fields (e.g., \texttt{MV}, \texttt{max\_MV}, \texttt{max\_MV\_node}), and a lazy tag \texttt{lazy}. We enforce:
\begin{itemize}
\item \texttt{height}$(\varnothing)=-1$ (so a leaf has height $0$).
\item \texttt{balance}$(u)=\texttt{height}(u.\text{left})-\texttt{height}(u.\text{right})$.
\item \texttt{PushDown}$(u)$ evaluates in $O(1)$ time and must be called before reading children or descending.
\item \texttt{Update}$(u)$ recomputes \texttt{height} and \emph{all} augmented fields of $u$ strictly from its children in $O(1)$ time.
\end{itemize}

\begingroup
\SetAlCapFnt{\footnotesize}
\SetAlCapNameFnt{\footnotesize}
\SetAlFnt{\footnotesize}
\SetAlgoNlRelativeSize{-1}

\paragraph{Rotations (strict lazy, parent pointers).}
The rotations reconnect the rotated subtree to its parent via parent pointers, and correctly update augmented fields from the bottom up.

\begin{algorithm}[H]
\caption{Utility: \texttt{RotateLeft}$(u)$ \quad (AVL, strict lazy, parent pointers)}\label{alg:rot-left}
\DontPrintSemicolon
\KwIn{Node $u$ with $u.\text{right}\neq\varnothing$}
\KwOut{New local root $y$ of the rotated subtree}
\SetKwProg{Fn}{Function}{:}{}
\Fn{\texttt{RotateLeft}$(u)$}{
  \texttt{PushDown}$(u)$\;
  $y \gets u.\text{right}$;\quad \texttt{PushDown}$(y)$\;
  $\beta \gets y.\text{left}$;\quad $p \gets u.\text{parent}$\;

  $y.\text{parent} \gets p$\;
  \uIf{$p\neq\varnothing$}{
    \uIf{$p.\text{left}=u$}{$p.\text{left}\gets y$\;}\Else{$p.\text{right}\gets y$\;}
  }

  $y.\text{left}\gets u$;\quad $u.\text{parent}\gets y$\;
  $u.\text{right}\gets \beta$\;
  \If{$\beta\neq\varnothing$}{$\beta.\text{parent}\gets u$\;}

  \texttt{Update}$(u)$;\ \texttt{Update}$(y)$\;
  \Return $y$\;
}
\end{algorithm}

\begin{algorithm}[H]
\caption{Utility: \texttt{RotateRight}$(u)$ \quad (AVL, strict lazy, parent pointers)}\label{alg:rot-right}
\DontPrintSemicolon
\KwIn{Node $u$ with $u.\text{left}\neq\varnothing$}
\KwOut{New local root $y$ of the rotated subtree}
\SetKwProg{Fn}{Function}{:}{}
\Fn{\texttt{RotateRight}$(u)$}{
  \texttt{PushDown}$(u)$\;
  $y \gets u.\text{left}$;\quad \texttt{PushDown}$(y)$\;
  $\beta \gets y.\text{right}$;\quad $p \gets u.\text{parent}$\;

  $y.\text{parent} \gets p$\;
  \uIf{$p\neq\varnothing$}{
    \uIf{$p.\text{left}=u$}{$p.\text{left}\gets y$\;}\Else{$p.\text{right}\gets y$\;}
  }

  $y.\text{right}\gets u$;\quad $u.\text{parent}\gets y$\;
  $u.\text{left}\gets \beta$\;
  \If{$\beta\neq\varnothing$}{$\beta.\text{parent}\gets u$\;}

  \texttt{Update}$(u)$;\ \texttt{Update}$(y)$\;
  \Return $y$\;
}
\end{algorithm}

\begin{algorithm}[H]
\caption{Utility: \texttt{RebalanceUp}$(u)$ \quad (standard AVL maintenance)}\label{alg:rebalance-up}
\DontPrintSemicolon
\KwIn{Node pointer $u$ (possibly $\varnothing$)}
\KwOut{Root of the rebalanced tree}
\SetKwProg{Fn}{Function}{:}{}
\Fn{\texttt{RebalanceUp}$(u)$}{
  \If{$u=\varnothing$}{\Return $\varnothing$\;}
  $w \gets u$\;
  \While{$w\neq\varnothing$}{
    \texttt{PushDown}$(w)$;\ \texttt{Update}$(w)$\;
    $b \gets \texttt{balance}(w)$\;

    \uIf{$b=2$}{
      \texttt{PushDown}$(w.\text{left})$;\ \texttt{Update}$(w.\text{left})$\;
      \If{$\texttt{balance}(w.\text{left})<0$}{
        \texttt{RotateLeft}$(w.\text{left})$\;
      }
      $w \gets \texttt{RotateRight}(w)$\;
    }
    \uElseIf{$b=-2$}{
      \texttt{PushDown}$(w.\text{right})$;\ \texttt{Update}$(w.\text{right})$\;
      \If{$\texttt{balance}(w.\text{right})>0$}{
        \texttt{RotateRight}$(w.\text{right})$\;
      }
      $w \gets \texttt{RotateLeft}(w)$\;
    }

    \texttt{Update}$(w)$\;
    \uIf{$w.\text{parent} = \varnothing$}{ \Return $w$\; }
    $w \gets w.\text{parent}$\;
  }
}
\end{algorithm}

\paragraph{Extracting extrema (returns the removed node as a singleton).}
These utilities isolate a pivot node required for concatenating trees of arbitrary heights.

\begin{algorithm}[H]
\caption{Utility: \texttt{MakeSingleton}$(x)$}\label{alg:singleton}
\DontPrintSemicolon
\KwIn{Node pointer $x$}
\KwOut{The same node $x$, completely isolated as a one-node AVL tree}
\SetKwProg{Fn}{Function}{:}{}
\Fn{\texttt{MakeSingleton}$(x)$}{
  $x.\text{left}\gets\varnothing$;\quad $x.\text{right}\gets\varnothing$;\quad $x.\text{parent}\gets\varnothing$\;
  $x.\text{lazy}\gets I$ \tcp*{Identity lazy tag}
  \texttt{Update}$(x)$\;
  \Return $x$\;
}
\end{algorithm}

\begin{algorithm}[H]
\caption{Utility: \texttt{ExtractMax}$(T)$ \quad (AVL, strict lazy)}\label{alg:extract-max}
\DontPrintSemicolon
\KwIn{AVL root $T$}
\KwOut{Pair $(x,T')$ where $x$ is the max node (singleton) and $T'$ is the remaining tree}
\SetKwProg{Fn}{Function}{:}{}
\Fn{\texttt{ExtractMax}$(T)$}{
  \If{$T=\varnothing$}{\Return $(\varnothing,\varnothing)$\;}
  $u \gets T$\;
  \While{$u.\text{right}\neq\varnothing$}{
    \texttt{PushDown}$(u)$;\ $u \gets u.\text{right}$\;
  }
  \texttt{PushDown}$(u)$;\ \texttt{Update}$(u)$\;

  $x \gets u$;\quad $p \gets x.\text{parent}$;\quad $c \gets x.\text{left}$ \tcp*{$x.\text{right}$ is $\varnothing$}
  \If{$c\neq\varnothing$}{$c.\text{parent}\gets p$\;}

  \uIf{$p=\varnothing$}{
    $T' \gets c$;\ \If{$T'\neq\varnothing$}{$T'.\text{parent}\gets\varnothing$\;}
  }
  \Else{
    \uIf{$p.\text{left}=x$}{$p.\text{left}\gets c$\;}\Else{$p.\text{right}\gets c$\;}
    $T' \gets \texttt{RebalanceUp}(p)$\;
  }
  $x \gets \texttt{MakeSingleton}(x)$\;
  \Return $(x,T')$\;
}
\end{algorithm}

\begin{algorithm}[H]
\caption{Utility: \texttt{ExtractMin}$(T)$ \quad (AVL, strict lazy)}\label{alg:extract-min}
\DontPrintSemicolon
\KwIn{AVL root $T$}
\KwOut{Pair $(x,T')$ where $x$ is the min node (singleton) and $T'$ is the remaining tree}
\SetKwProg{Fn}{Function}{:}{}
\Fn{\texttt{ExtractMin}$(T)$}{
  \If{$T=\varnothing$}{\Return $(\varnothing,\varnothing)$\;}
  $u \gets T$\;
  \While{$u.\text{left}\neq\varnothing$}{
    \texttt{PushDown}$(u)$;\ $u \gets u.\text{left}$\;
  }
  \texttt{PushDown}$(u)$;\ \texttt{Update}$(u)$\;

  $x \gets u$;\quad $p \gets x.\text{parent}$;\quad $c \gets x.\text{right}$ \tcp*{$x.\text{left}$ is $\varnothing$}
  \If{$c\neq\varnothing$}{$c.\text{parent}\gets p$\;}

  \uIf{$p=\varnothing$}{
    $T' \gets c$;\ \If{$T'\neq\varnothing$}{$T'.\text{parent}\gets\varnothing$\;}
  }
  \Else{
    \uIf{$p.\text{left}=x$}{$p.\text{left}\gets c$\;}\Else{$p.\text{right}\gets c$\;}
    $T' \gets \texttt{RebalanceUp}(p)$\;
  }
  $x \gets \texttt{MakeSingleton}(x)$\;
  \Return $(x,T')$\;
}
\end{algorithm}
\begin{algorithm}[H]
\caption{\texttt{JoinWithPivot}$(T_1,x,T_2)$ \quad (Safely joins trees of arbitrary heights)}\label{alg:join-with-pivot}
\DontPrintSemicolon
\KwIn{AVL roots $T_1,T_2$; singleton node $x$ bridging their key domains}
\KwOut{AVL root of $T_1 \cup \{x\} \cup T_2$}

\SetKwProg{Fn}{Function}{:}{}
\Fn{\texttt{JoinWithPivot}$(T_1, x, T_2)$}{
  \tcp{Safely obtain heights (treating empty trees as height -1)}
  $h_1 \gets -1$;\quad $h_2 \gets -1$\;
  \If{$T_1\neq\varnothing$}{\texttt{PushDown}$(T_1)$;\ \texttt{Update}$(T_1)$;\ $h_1 \gets \text{height}(T_1)$\;}
  \If{$T_2\neq\varnothing$}{\texttt{PushDown}$(T_2)$;\ \texttt{Update}$(T_2)$;\ $h_2 \gets \text{height}(T_2)$\;}

  \BlankLine
  \uIf{$|h_1-h_2|\le 1$}{
    $x.\text{left}\gets T_1$;\ \If{$T_1\neq\varnothing$}{$T_1.\text{parent}\gets x$\;}
    $x.\text{right}\gets T_2$;\ \If{$T_2\neq\varnothing$}{$T_2.\text{parent}\gets x$\;}
    $x.\text{parent}\gets\varnothing$\;
    \texttt{Update}$(x)$;\ \Return $x$\;
  }
  \uElseIf{$h_1 > h_2 + 1$}{
    $u \gets T_1$\;
    \While{\textbf{true}}{
      \texttt{PushDown}$(u)$;\ \texttt{Update}$(u)$\;
      $v \gets u.\text{right}$\;
      $h_v \gets -1$\;
      \If{$v\neq\varnothing$}{\texttt{PushDown}$(v)$;\ \texttt{Update}$(v)$;\ $h_v \gets \text{height}(v)$\;}
      \If{$h_v \le h_2+1$}{\textbf{break}\;}
      $u \gets v$\;
    }
    $v \gets u.\text{right}$ \tcp*{may be $\varnothing$}
    $u.\text{right}\gets x$;\ $x.\text{parent}\gets u$\;
    $x.\text{left}\gets v$;\ \If{$v\neq\varnothing$}{$v.\text{parent}\gets x$\;}
    $x.\text{right}\gets T_2$;\ \If{$T_2\neq\varnothing$}{$T_2.\text{parent}\gets x$\;}
    \texttt{Update}$(x)$;\ \texttt{Update}$(u)$\;
    \Return \texttt{RebalanceUp}(u)\;
  }
  \Else{
    $u \gets T_2$\;
    \While{\textbf{true}}{
      \texttt{PushDown}$(u)$;\ \texttt{Update}$(u)$\;
      $v \gets u.\text{left}$\;
      $h_v \gets -1$\;
      \If{$v\neq\varnothing$}{\texttt{PushDown}$(v)$;\ \texttt{Update}$(v)$;\ $h_v \gets \text{height}(v)$\;}
      \If{$h_v \le h_1+1$}{\textbf{break}\;}
      $u \gets v$\;
    }
    $v \gets u.\text{left}$ \tcp*{may be $\varnothing$}
    $u.\text{left}\gets x$;\ $x.\text{parent}\gets u$\;
    $x.\text{right}\gets v$;\ \If{$v\neq\varnothing$}{$v.\text{parent}\gets x$\;}
    $x.\text{left}\gets T_1$;\ \If{$T_1\neq\varnothing$}{$T_1.\text{parent}\gets x$\;}
    \texttt{Update}$(x)$;\ \texttt{Update}$(u)$\;
    \Return \texttt{RebalanceUp}(u)\;
  }
}
\end{algorithm}

\begin{algorithm}[H]
\caption{\texttt{Join}$(T_1,T_2)$ \quad (Requires all keys in $T_1 <$ all keys in $T_2$)}\label{alg:join-fixed}
\DontPrintSemicolon
\KwIn{AVL roots $T_1,T_2$ with $\max(T_1)<\min(T_2)$}
\KwOut{AVL root containing $T_1\cup T_2$}
\SetKwProg{Fn}{Function}{:}{}
\Fn{\texttt{Join}$(T_1, T_2)$}{
  \If{$T_1=\varnothing$}{\Return $T_2$\;}
  \If{$T_2=\varnothing$}{\Return $T_1$\;}

  \texttt{PushDown}$(T_1)$;\ \texttt{Update}$(T_1)$;\quad \texttt{PushDown}$(T_2)$;\ \texttt{Update}$(T_2)$\;

  \eIf{$\text{height}(T_1)\ge \text{height}(T_2)$}{
    $(x,T_1') \gets \texttt{ExtractMax}(T_1)$\;
    \Return \texttt{JoinWithPivot}$(T_1',x,T_2)$\;
  }{
    $(x,T_2') \gets \texttt{ExtractMin}(T_2)$\;
    \Return \texttt{JoinWithPivot}$(T_1,x,T_2')$\;
  }
}
\end{algorithm}

\begin{algorithm}[H]
\caption{\texttt{Split}$(T,k_{\text{split}})$ \quad (Strict invariant preservation via pivoting)}\label{alg:split-fixed}
\DontPrintSemicolon
\KwIn{AVL root $T$, split key $k_{\text{split}}$}
\KwOut{$(T_{\text{left}},T_{\text{right}})$ with all keys in $T_{\text{left}}<k_{\text{split}}\le$ all keys in $T_{\text{right}}$}

\SetKwProg{Fn}{Function}{:}{}
\Fn{\texttt{Split}$(T,k_{\text{split}})$}{
  \If{$T=\varnothing$}{\Return $(\varnothing,\varnothing)$\;}
  
  \tcp{Push tags down to children before dismantling the node}
  \texttt{PushDown}$(T)$;\ \texttt{Update}$(T)$\;
  
  \tcp{1. Decouple children to create valid, independent subtrees}
  $L_{\mathrm{orig}} \gets T.\text{left}$;\quad $R_{\mathrm{orig}} \gets T.\text{right}$\;
  \If{$L_{\mathrm{orig}}\neq\varnothing$}{$L_{\mathrm{orig}}.\text{parent}\gets\varnothing$\;}
  \If{$R_{\mathrm{orig}}\neq\varnothing$}{$R_{\mathrm{orig}}.\text{parent}\gets\varnothing$\;}
  
  \tcp{2. Isolate the root to act as a clean pivot node}
  $T \gets \texttt{MakeSingleton}(T)$\;

  \tcp{3. Recursively split the appropriate child and merge safely using the pivot}
  \uIf{$k_{\text{split}} \le T.\text{key}$}{
    $(L, M) \gets \texttt{Split}(L_{\mathrm{orig}}, k_{\text{split}})$\;
    $T_{\text{right}} \gets \texttt{JoinWithPivot}(M, T, R_{\mathrm{orig}})$\;
    \Return $(L, T_{\text{right}})$\;
  }
  \Else{
    $(M, R) \gets \texttt{Split}(R_{\mathrm{orig}}, k_{\text{split}})$\;
    $T_{\text{left}} \gets \texttt{JoinWithPivot}(L_{\mathrm{orig}}, T, M)$\;
    \Return $(T_{\text{left}}, R)$\;
  }
}
\end{algorithm}
\endgroup

\paragraph{Time complexity.}
The strict bounding of AVL properties natively absorbs our mathematical augmentations. \texttt{PushDown} and \texttt{Update} evaluate in $O(1)$ time by strictly limiting touches to a parent and its immediate children. Consequently, incorporating them into a standard root-to-leaf traversal does not impact the basic search complexity of $O(\log n)$. Bulk lazy tags (\texttt{AddHoldingCosts}, \texttt{BulkUpdateQAtP2}) perform an $O(1)$ constant-time write at the local root of a partition. 

For the structural operations, \texttt{JoinWithPivot} bridges subtrees of disparate heights by walking down the spine of the taller tree. The walk requires $O(|h_1 - h_2|)$ bounds. When \texttt{Split} is called, it dismantles the tree along a single path of length $O(\log n)$. As the recursion unwinds, \texttt{JoinWithPivot} is used repeatedly. Crucially, the height differences bridged by \texttt{JoinWithPivot} across the recursion form a \textbf{telescoping sum}. Since the sum of height differences along any root-to-leaf path is strictly bounded by the total height of the original tree, the aggregate time spent joining fragments during a split evaluates strictly to $O(h_{\text{root}})$. Therefore, both \texttt{Split} and \texttt{Join} execute in $O(\log n)$ worst-case time (see Brass \cite{brass2008advanced}, pp.~145--147, for formal proofs on AVL splits via telescoping bounds).
\begin{algorithm}[H]
\caption{Baseline dynamic program (naive)}\label{alg:naive}
\DontPrintSemicolon
$\DP(t,c)$ represents the minimum cost to reach period $t$ and end period $t$
(within capacity) with $c$ units of fuel/inventory remaining.\;
We assume that if a state from a previous period is still optimal at the current period,
it is kept even if there is another way to generate the same optimal value.\;
We also assume that when $p_{1,t}=p_{2,t}$, the discounted tier $p_{2,t}$ is preferred whenever feasible.\;

\BlankLine
For every period $t \in \{1,\ldots,n\}$ and every feasible fuel level $0 \le c \le B(t)$:
\[
\DP(t,c)
=
\min_{\substack{0 \le x_t \\ 0 \le c' \le B(t-1) \\ c' + x_t - d_t = c}}
\left(\DP(t-1,c') + p_t(x_t) + h_t\,c\right).
\]
with initial conditions:
\[
\DP(0,0)=0,\qquad \DP(0,c\neq 0)=+\infty.
\]
where the pricing function is:
\[
p_t(x_t)=
\begin{cases}
p_{1,t}\,x_t, & x_t < Q,\\[4pt]
p_{2,t}\,x_t, & x_t \ge Q \quad (p_{2,t}\le p_{1,t}).
\end{cases}
\]
A concrete ordering decision $(c',x_t)$ is characterized by the period $t$ and the price tier
$\alpha \in \{1\text{ (using }p_{1,t}),\,2\text{ (using }p_{2,t})\}$.\;

\BlankLine
\textbf{Note:} The constraint $c' + x_t - d_t = c$ means that to end period $t$ with $c$ units
(after satisfying demand $d_t$), we start with $c'$ units and order $x_t$ units.\;
\end{algorithm}
\end{document}